\documentclass[11pt]{article}

\usepackage[preprint]{acl}

\usepackage{times}
\usepackage{latexsym}
\usepackage[T1]{fontenc}
\usepackage[utf8]{inputenc}
\usepackage{microtype}
\usepackage{inconsolata}
\usepackage{graphicx}

\usepackage{amsmath,amssymb,amsthm}
\usepackage{booktabs}
\usepackage{multirow}
\usepackage{array}
\usepackage{enumitem}
\usepackage{algorithm}
\usepackage{algorithmic}
\setlist{leftmargin=*, itemsep=0pt}
\setlist[itemize,1]{leftmargin=10pt}
\renewcommand{\topfraction}{.95}
\renewcommand{\dbltopfraction}{.95}
\renewcommand{\textfraction}{.05}
\newtheorem{theorem}{Theorem}

\theoremstyle{definition}

\theoremstyle{remark}
\newtheorem{remark}{Remark}

\DeclareMathOperator*{\argmax}{arg\,max}

\title{Pin It on Anyone: Certified Targeted Displacement of\\Multi-Bit LLM Watermarks Without the Key}

\author{%
  Pawe\l{} Batorski\textsuperscript{1}\quad
  Przemys\l{}aw Spurek\textsuperscript{2,3}\quad
  Paul Swoboda\textsuperscript{1}\\[3pt]
  \normalfont\normalsize\textsuperscript{1}Heinrich Heine University D\"usseldorf\quad
  \textsuperscript{2}Jagiellonian University\quad
  \textsuperscript{3}IDEAS Research Institute}

\begin{document}
\maketitle

\begin{abstract}
Multi-bit watermarking embeds a message, typically an account identifier, into LLM output so that harmful text can be traced to its origin.  
We show an attack technique that allows to alter the watermarking message.
This allows to, for example, exchange the embedded account ID to an arbitrary different one. 
The attacker mainly needs access to a few accounts for generating LLM output and model logits, but crucially not the secret key used to watermark.
The attacker then rewrites the text, keeping its meaning, so that the watermarking message decoder recovers the watermarking message embedded by the attacker.
We are, to our knowledge, the first to propose such an attack, which we call a \emph{targeted displacement}.
Theoretically, we present upper and lower bounds on the probability of our attack succeeding.
This bound depends on the number of accounts the attacker has access to, the length of the message and some characteristics of the given text.
Our attack can succeed whenever we can bound the ambiguity of the generated text well enough.
Empirically, we test our attack technique against five watermarking methods and four LLMs.
Attack success tracks our predicted bounds, allowing us to alter watermarking messages in almost all cases with high accuracy, given enough accounts.
Our code can be found at \url{https://github.com/Batorskq/PinIt}.
\end{abstract}

\section{Introduction}

\begin{figure}[t!]
 \centering
 \includegraphics[width=\columnwidth]{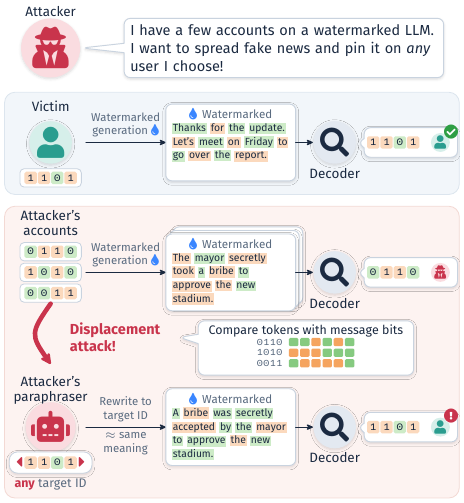}
 \caption{Targeted displacement of a multi-bit watermark.  The provider embeds
 each account's identifier into generated text, and a decoder recovers it to
 attribute a text.  Text from the attacker's own accounts is traced back to the
 attacker.  By comparing the tokens its accounts favor with their known message
 bits, the attacker learns the hidden watermark structure without the key and
 paraphrases its text, keeping the meaning, so that the decoder outputs any
 identifier of the attacker's choice, here that of an innocent user.}
 \label{fig:teaser}
\end{figure}

Watermarks let a provider recognize text that its language model generated \citep{kirchenbauer2023watermark}.
Multi-bit watermarks also hide a message in the text, typically the identifier of the account that requested it, so that a decoder holding the secret key can attribute a harmful text to its account \citep{fernandez2023three,yoo-etal-2024-advancing,qu-etal-2025-provably,li-etal-2024-identifying,jiang-etal-2025-stealthink,xu-etal-2026-xmark}.
We show that an attacker can effectively alter the watermark message (Figure~\ref{fig:teaser}).
With a few accounts of its own and without the secret key, it rewrites a text that carries a \emph{source message $m$}, keeping its semantics, so that the decoder reads a chosen \emph{target message $m'$}, such as the identifier of a third user.
We call this \emph{targeted displacement}.
Obviously, such a watermarking attack can have serious security and privacy implications, when pinning harmful or sensitive text on an innocent party.

Prior attacks remove the watermark \citep{krishna2023paraphrasing,zhang2024watermarks} or spoof its presence \citep{jovanovic2024watermark,gu2024learnability,pang2024no,jiang-etal-2025-stealthink}.
They hence change whether the text is recognized as watermarked or not.
That a multi-bit watermark might be steered toward a specific user was anticipated \citep{jovanovic2024watermark,gloaguen2025discovering} and stated as a security goal \citep{qu-etal-2025-provably}, but no known attack, to our knowledge, was able to alter an underlying watermarking message.
We develop the first such attack and a theory of what any such attack can achieve.

We study attacks on so called distortion-based watermarking methods, which include popular green-list and red-list methods \citep{kirchenbauer2023watermark,kirchenbauer2024reliability,zhao2024provable,lee-etal-2024-wrote,liu2024semantic,lu-etal-2024-entropy} and their multi-bit variants \citep{fernandez2023three,yoo-etal-2024-advancing,wang2024codable,qu-etal-2025-provably}.
Such schemes make the generator favor, at each step, a subset of the vocabulary that depends on the message and the secret key.
The decoder checks how consistently a text follows the subsets of each candidate message.
The attacker cannot compute these subsets without the key, but a few accounts whose messages are known can reveal enough information about the favored subset.
When comparing an account's next-token distribution with the base model, we see which tokens are favored, and comparing accounts with one another shows the tokens the target message $m'$ would favor.
A paraphraser restricted to these tokens then rewrites the text (Figure~\ref{fig:method}).
Not every target message is reachable however: when the accounts are too few or their underlying watermark messages are too similar, several targets look the same to the attacker.
For such cases we prove an upper bound on the efficacy of any possible attacker.
Additionally, we provide lower bounds for our attacker: 
For any target message $m'$ can compute a certificate, without access to the key, for a lower bound on the probability that the decoder will decode $m'$.
We evaluate the attack on five published watermarking methods using four LLMs
(Section~\ref{sec:experiments}).

Our contributions are the following.
\begin{itemize}
\item To our knowledge, we develop the first attack on a multi-bit watermarking method that can inject arbitrary target messages without access to the key.
 Its three steps depend on the scheme.
\item We give upper bounds on attack success for any attack method on multi-bit watermarking schemes and a lower bound that certifies the probability of attack success for our method. We exactly characterize in closed form reachable target messages.
\item We evaluate the attack on five published watermarking methods and four LLMs, confirm the bound, show which schemes fall to one account, to multiple colluding accounts, or resist, and measure the price in fidelity for generating watermarked text with a target message.
\end{itemize}

\section{Related Work}
\label{sec:related}

\paragraph{Multi-bit watermarks.}
Zero-bit watermarks bias the next-token distribution toward a keyed token
set \citep{chen-etal-2024-watme,huo2024token,kirchenbauer2023watermark,kirchenbauer2024reliability,lee-etal-2024-wrote,liu2024semantic,liu2024adaptive,lu-etal-2024-entropy,wang-etal-2025-morphmark,zhao2024provable},
sample from it without distortion
\citep{chen-etal-2025-improved,christ2024undetectable,fu-etal-2024-gumbelsoft,golowich2024edit,he2025theoretically,hu2024unbiased,mao-etal-2025-watermarking,tsur2025heavywater,wu2024resilient,zhao2025permute},
or select among drafts \citep{giboulot2024watermax,hou-etal-2024-semstamp}.
They only detect presence.  Multi-bit watermarks, like linguistic
steganography \citep{shen-etal-2020-near,ziegler-etal-2019-neural}, hide a
message, typically a user identifier.  CycleShift cyclically shifts one keyed
set by the message \citep{fernandez2023three}, MPAC gives each context a
message position and a vocabulary coloring \citep{yoo-etal-2024-advancing},
RSBH spreads a Reed--Solomon codeword over public segments
\citep{qu-etal-2025-provably}, DepthW keys a depth permutation by the whole
code \citep{li-etal-2024-identifying}, StealthInk reweights a keyed
permutation without distortion \citep{jiang-etal-2025-stealthink}, and XMark
suppresses keyed vocabulary quarters within public blocks
\citep{xu-etal-2026-xmark}.  Other schemes edit or paraphrase the text
\citep{lau-etal-2024-waterfall,xu2025robust,yang2022tracing,yoo-etal-2023-robust},
balance, stack, or symmetrize partitions
\citep{feng2025bimark,park2026watermod,wang2024codable,zhu-etal-2026-quantilemark},
correct, verify, or layer the payload
\citep{chen2026selective,kim2026blockwise,li2026argh}, personalize it
\citep{an2026ipmark,yu2025saemark}, or serve diffusion LMs, agents, and code
models
\citep{huang-wan-2026-tdma,huang-etal-2026-agentmark,yang-etal-2026-second,zhang-etal-2025-clmtracing}.
Attribution has been formalized \citep{jiang2026watermark,song2026ideal}, and
CertMark certifies the decoded message \citep{batorski2026certmark}.  These
works protect the message against erasure by edits and paraphrases.  We ask
which messages an attacker can make the decoder read.

\paragraph{Attacks on watermarks.}
Paraphrasing, editing, translation, and ensembling remove watermarks
\citep{cheng2025revealing,cheng2025adversarial,diaa2025optimizing,harel-canada-etal-2025-sandcastles,he-etal-2024-watermarks,huang-etal-2025-b4,huang2026rlcracker,hwang2026llm,krishna2023paraphrasing,pan-etal-2025-llm,rastogi-pruthi-2024-revisiting,wu2026linear,zhang2024watermarks},
benchmarks compare schemes
\citep{pan-etal-2024-markllm,tu-etal-2024-waterbench,wu2026analyzing}, and
queries reveal whether and how a model is watermarked
\citep{gloaguen2025blackbox,liu2025can}.  Closer to us, stealing estimates a
zero-bit scheme's keyed sets from queries
\citep{chen2025demark,jovanovic2024watermark,wu-chandrasekaran-2024-bypassing}
and distillation copies a watermark into another model
\citep{an-etal-2026-ditto,gu2024learnability,sander2024watermarking}, so that
an attacker can scrub it or spoof it on any text.  Robustness alone permits
piggyback spoofing \citep{pang2024no}.  Defenses harden watermarks against
spoofing or detect it
\citep{an2025defending,an-etal-2026-reinforcement,gloaguen2025discovering,liu2024unforgeable,shen2025enhancing,zhou2024bileve},
and image watermarks can likewise be forged or overwritten
\citep{dong2025wmcopier,muller2025blackbox,saberi2024robustness,senthuran2026self,soucek2025transferable,yang2024can}.
A zero-bit spoof forges one verdict, whereas spoofing a multi-bit watermark could
incriminate a specific user \citep{gloaguen2025discovering}, and we study
this attack.  Combining outputs under several keys removes a watermark
\citep{pang2024no,wu2026linear}.  Our attacker combines accounts under one key
to forge a message, and our symmetry bound limits it as impossibility and
trade-off results limit the defender
\citep{pratihar2026catch,zhang2024watermarks}.  The rewrite steers a
paraphraser toward a token set, as in lexically constrained and controlled
decoding
\citep{dathathri2020plug,hu-etal-2019-improved,kajiwara-2019-negative,liu-etal-2021-dexperts,lu-etal-2021-neurologic,mudgal2024controlled,qin2022cold,yang-klein-2021-fudge},
also used to obfuscate authorship \citep{fisher-etal-2024-jamdec}.

\section{Method}
\label{sec:method}

\begin{figure*}[t]
 \centering
 \includegraphics[width=.999\textwidth]{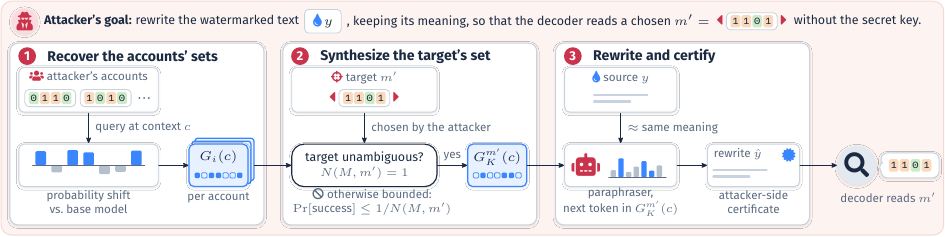}
 \caption{Overview of targeted displacement.  The attacker wants to rewrite a
 text that carries watermark message $m$, keeping its meaning, so that the
 decoder attributes it to a chosen message $m'$, without the secret key.  It
 controls accounts with known messages and, at every context used by the
 rewrite, compares an account's next-token distribution with the base model to
 recover the account's green list (step 1).  The recovered green lists are
 turned into the target's green list.  This succeeds exactly when the target's
 relative bits are constant within every twin class; otherwise the ambiguity
 bound of Theorem~\ref{thm:targeting-bound} caps the success probability
 (step 2).  A
 paraphraser rewrites the source under a token-level constraint toward the
 target's set, and the attacker certifies the rewrite from its own posterior
 without the key (step 3), so that the provider's decoder reads the target
 message.  Only step 2 depends on the watermark scheme.}
 \label{fig:method}
\end{figure*}

We study \emph{targeted displacement}: rewriting a text that carries a watermark message $m$ so that the decoder instead recovers the target message $m'$.
We first state the attacker's objective and describe the watermarks we study.

\paragraph{Watermarking (Attacks).}
\label{sec:problem-formulation}

Let $\mathcal M = \{0,1\}^L$ be the message space.
Let $p_{K,m}$ be the next-token distribution for encoding watermarked text with message $m$ and secret key $K$.
Typically $m$ will be the account identifier of the person generating text.
For text $y \sim p_{K,m}$ the decoder $D_K$ will output with high probability that $D_K(y) = m$.

The attacker's goal is as follows: 
For a target message $m'\in\mathcal M$, the attacker outputs text $\hat y=A(y,m',O_M)$ such that 
(i)~the decoder will output $m'$ as watermarking message and (ii)~the rewritten text $\hat{y}$ is semantically close to the original text $y$ and (iii) $O_M$ contains all the information the attacker has at his disposal.
First, the attacker has access to a number of accounts with messages 
\begin{equation}
 M=(m^{(1)},\ldots,m^{(n)})\in\mathcal M^n,\qquad n\geq1,
 \label{eq:known-messages}
\end{equation}
so he can generate text for $p_{K,m^{(1)}}$, $\ldots$, $p_{K,m^{(n)}}$.
Furthermore, the attacker has access to the base model, the tokenizer, the watermarking scheme, and relevant hyperparameters (decoding temperature, etc.) but has no access to the secret key $K$ nor to the decoder for retrieving the watermark message from text. 

The threat is direct for open-weight models, where every user of a watermarked service can download the base model.
For closed models it becomes an insider threat.

\paragraph{Distortion Watermarking.}
\label{sec:preliminaries}
We study \emph{distortion watermarks}, which, for a fixed key, embed the message by changing the model's next-token distribution \citep{kirchenbauer2023watermark}.
When generating text $y$, at every position $t$ the method takes the last $\ell$ tokens (the context length) of $y$, i.e.\ $c = y_{t-\ell:t}$.
Together with the key $K$, the message $m\in\mathcal M$, the \emph{green list} of tokens is computed:
\begin{equation}
 G_K^m(c)\subseteq\mathcal V,
 \label{eq:label-set}
\end{equation}
The green list is used as follows:
\begin{enumerate}[label=\textnormal{(P\arabic*)},leftmargin=*,labelsep=0.3em]
 \item generation under message $m$ replaces the base distribution
 $p(\cdot\mid c)$ by $p_{K,m}(\cdot|c) = \Phi\bigl(p(\cdot\mid c),\,G_K^m(c)\bigr)$ for a fixed,
 key-independent rule $\Phi$, and
\item the decoder recomputes the green lists for all messages $m \in \mathcal M$ and chooses the message that maximizes the probability of having generated the watermarked text $y$, i.e.\ $m \in \argmax{\tilde{m} \in \mathcal M} p_{K,\tilde{m}}(y)$.
\end{enumerate}
The rule $\Phi$ usually shifts probability toward the green list. 
Alternatively, some schemes such as XMark suppress the green list set instead.
We can also express this by taking $G_K^m(c)$ to be the complement of the original green set (Appendix~\ref{app:transformations}).
For simplicity, for a controlled account $i$ with known message $m^{(i)}$, we write $G_i(c):=G_K^{m^{(i)}}(c)$.

\paragraph{Green-List $G_K^m$ Construction.}
The watermarking method uses key $K$ and context $c$ to compute a pseudorandom permutation $\pi_K(c)$ of the vocabulary and splits the permuted tokens into two halves.
This assigns every token $v$ a hidden color $C_K(c,v)\in\{0,1\}$.
The key also determines a message position $J_K(c)\in[L]$.
For message $m$, the green list is computed as those tokens whose color agrees with the message $m$ at position  $J_K(c)$, namely 
\begin{equation}
 G_K^m(c)=\{v\in\mathcal V:C_K(c,v)=m_{J_K(c)}\}.
 \label{eq:binary-label-set}
\end{equation}
The attacker has no access to the key, hence cannot compute the permutation $\pi_K(c)$, nor message position $J_K$ or hidden color $C_K$.
It observes only the resulting probability shift.
Figure~\ref{fig:hidden-set} shows why $G_K^m(c)$ is therefore the relevant object.

\begin{figure}[!ht]
 \centering
 \includegraphics[width=.969\columnwidth]{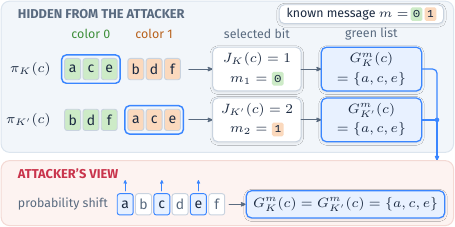}
 \caption{Different hidden permutations and selections can induce the same
 observable green list set.}
 \label{fig:hidden-set}
\end{figure}

\subsection{Attack}
\label{sec:attack}

\paragraph{Recovering $G_i(c)$.}
\label{sec:recovery}
Fix any context $c$.
Let $p(\cdot\mid c)$ be the base model's next-token distribution, and let $p_i(\cdot\mid c):=p_{K,m^{(i)}}(\cdot\mid c)$ be the observable distribution of controlled account $i$.
For general rules $\Phi$, the true green list set $G_i(c)$ is unknown, since the attacker does not know $K$.
However, for all the used watermarking methods function $\Phi$ is simple and we can obtain $G_i(c)$ analytically.
For example, for MPAC $\Phi$ increases the probability for all tokens in $G_i(c)$ and reduces them for the complement.
Hence $G_i(c) = \left\{v \in \mathcal V : p(v|c) < p_i(v|c) \right\}$.

We present a general inversion approach that computes $G_i(c)$ consistent with arbitrary $\Phi$ and give analytic inversion formulas for every used watermarking method in Appendix~\ref{app:transformations}.

\begin{figure}[t]
\centering
\includegraphics[width=.89\columnwidth]{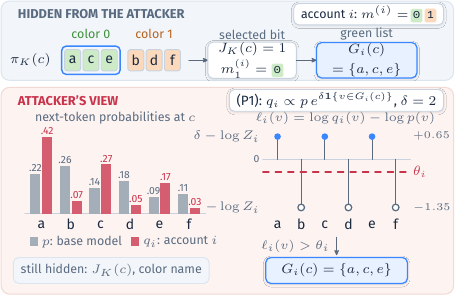}
\caption{Recovering an account's set for MPAC at one context $c$.  Top: the
key colors the vocabulary and selects position~1, where account $i$'s
message has bit~0, so generation boosts $G_i(c)=\{a,c,e\}$.  The attacker
sees none of this.  Bottom: the attacker queries the account and the base
model at $c$.  By the public rule (P1), every token of the hidden set has
its probability multiplied by $e^{\delta}$ before renormalization, so the
log-ratio $\ell_i(v)=\log q_i(v)-\log p(v)$ takes only two values,
$\delta-\log Z_i$ inside the set and $-\log Z_i$ outside it, here with
$Z_i=1+(e^{\delta}-1)\sum_{v\in G_i(c)}p(v)=3.9$.  The midpoint $\theta_i$
separates them and returns $\{a,c,e\}$ exactly, without the key.  The
position $J_K(c)$ and the name of the color stay hidden.
Appendix~\ref{app:transformations} gives the rule for every scheme.}
\label{fig:recover-set}
\end{figure}

\paragraph{Recovering green list of target $m'$.}
\label{sec:synthesis}
The attacker controls $n$ accounts and, at context $c$, holds their green lists $G_1(c),\ldots,G_n(c)$.
We detail below how the attacker uses them to compute the target set  $G_K^{m'}(c)$.
Recall that in \eqref{eq:binary-label-set} the key assigns $c$ the position $j=J_K(c)$, which the attacker does not know.
It knows, however, that $G_i(c)$ is the class of tokens with color $m_j^{(i)}$, and that the two color classes partition $\mathcal V$.
Two accounts therefore have the same green list at $c$ exactly when their messages have the same bit at the hidden position,
\begin{equation}
 G_i(c)=G_{i'}(c)
 \iff
 m_j^{(i)}=m_j^{(i')},
 \label{eq:same-set}
\end{equation}
and complementary sets otherwise.

To relate these comparisons to positions, we take account $1$ as the reference and define, for every message position $j\in[L]$, its
\emph{account signature}
\begin{equation}
 \begin{aligned}
 x_j&=(x_{j,2},\ldots,x_{j,n})\in\{0,1\}^{n-1},\\[-2pt]
 x_{j,i}&=\bigl(m_j^{(i)}+m_j^{(1)}\bigr)\bmod 2,
 \end{aligned}
 \label{eq:account-signature}
\end{equation}
so $x_{j,i}=0$ if account $i$ has the same bit as the reference at position $j$, and $x_{j,i}=1$ if it has the opposite bit.
The attacker computes every $x_j$ from the known messages, and by \eqref{eq:same-set} the sets reveal the signature of the hidden position, $x_{J_K(c),i}=\mathbf 1\{G_i(c)\neq G_1(c)\}$.
Positions $j$ and $j'$ are
\emph{twins} if $x_j=x_{j'}$.  We write $\mathcal T$ for the resulting
partition of $[L]$ into \emph{twin classes} and $T(c)$ for the class of
$J_K(c)$.  Matching the observed signature against the known ones identifies
$T(c)$, but no comparison of sets can separate twins
(Figure~\ref{fig:twin-classes}).

\begin{figure}[t]
\centering
\includegraphics[width=.89\columnwidth]{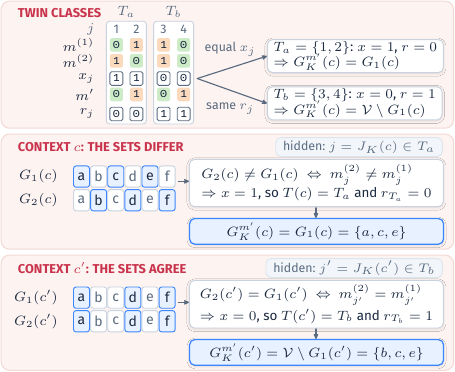}
\caption{Building the set of the target $m'$ from twin classes.  Top: two
accounts with known messages, $m^{(1)}$ as the reference, and the signature
$x_j$ that marks where $m^{(2)}$ differs from it.  Positions with equal
$x_j$ are twins, here $T_a=\{1,2\}$ and $T_b=\{3,4\}$.  The target $m'$
differs from the reference by $r_j$, which is constant on each class, so
every class fixes one choice between $G_1(c)$ and its complement
\eqref{eq:target-set}.  Middle: at a context $c$ the two accounts' green
lists differ, which by \eqref{eq:same-set} means that their messages differ
at the hidden position $j=J_K(c)$, so $j\in T_a$ and the attacker boosts
$G_1(c)$.  Bottom: at a context $c'$ the sets agree, so $j'\in T_b$ and the
attacker boosts the complement of $G_1(c')$.  Neither step needs the position
or the names of the colors.  Bit 0 is green, bit 1 is orange, and the set to
boost is blue.}
\label{fig:twin-classes}
\end{figure}

Let $r_j=(m'_j+m_j^{(1)})\bmod 2$ be the target's bit relative to the reference.
Since $G^{m'}_K(c)$ is the class of tokens with color $m'_{J_K(c)}$, it equals $G_1(c)$ if $r_{J_K(c)}=0$ and $\mathcal V\setminus G_1(c)$ otherwise.
If $r$ is constant on every twin class, with value $r_T$ on class $T$, the attacker can make this choice from $T(c)$ alone:
\begin{equation}
 G^{m'}_K(c)=
 \begin{cases}
  G_1(c), & r_{T(c)}=0,\\[-1pt]
  \mathcal V\setminus G_1(c), & r_{T(c)}=1.
 \end{cases}
 \label{eq:target-set}
\end{equation}
The attacker never needs the hidden position or the name of either color.
This covers $2^{|\mathcal T|}$ targets, one per choice of the values $r_T$
(all $2^L$ when every position has a distinct signature).  It remains to use
these sets to steer the rewrite.

\paragraph{Redirecting the decoded message to $m'$.}
\label{sec:constructive-attack}
Let $y$ be the source text.
We will use a paraphraser model $p_{para}$, a chat model that is asked to rewrite $y$ faithfully (Appendix~\ref{app:prompt}).
Let $z$ be the rewritten text and $z_{<t}$ the rewrite prefix at position $t$.
We write
\[
 p_t(v):=p_{\mathrm{para}}(v\mid y,z_{<t})
\]
for the paraphraser's next-token distribution.
The prompt never mentions the watermark, which enters only through the decoding rule below.
At the watermark context $z_{<t}$ we have the green-list $G^{(m')}_K(z_{<t})$.
We use the \emph{hard bias sampler}
\begin{equation}
 r_t(v) \propto \begin{cases} p_t(v),& v \in G^{(m')}_K(z_{<t}) \\ 0,& \text{otherwise} \end{cases}.
 \label{eq:hard-attack}
\end{equation}
We apply \eqref{eq:hard-attack} when $Z_t\geq\tau$ and $G^{(m')}_K(z_{<t})$ contains an available token.
Otherwise we retain the full vocabulary and use a \emph{soft bias sampler}.
Choose a $\delta_{\mathrm{atk}}>0$ and let
\begin{equation}
 r_t(v)\propto
 p_t(v)\exp\!\left(\delta_{\mathrm{atk}}\mathbf 1\{v\in G^{(m')}_K(z_{<t}\}\right).
 \label{eq:soft-attack}
\end{equation}
Thus $\tau=0$ forces a target-compatible token whenever possible, while larger $\tau$ leaves low-mass contexts closer to the paraphraser.

The decoder relies on distinct context-token pairs for its efficacy. 
Since context-token pairs can repeat too often (recall that context-token pairs can reappear due to fixed context length $\ell$), and constrained decoding tends to loop, we block repeated 4-grams \citep{paulus2018deep}.
Algorithm~\ref{alg:mpac-attack} in Appendix~\ref{app:pseudocode} lists the complete attack for binary MPAC.

\begin{figure*}[t]
\centering
\includegraphics[width=\textwidth]{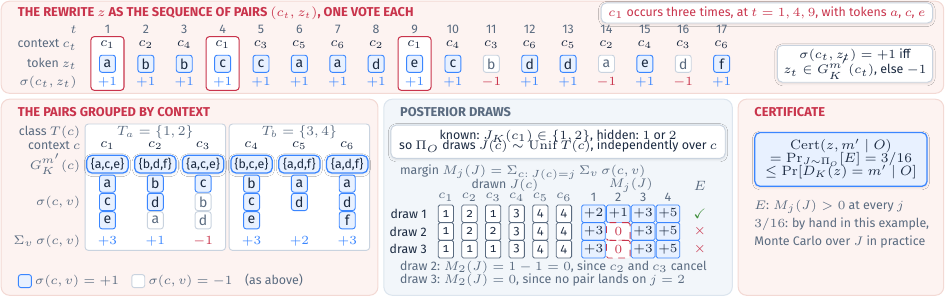}
\caption{The certified attack on a rewrite $z$, for $L=4$ with twin classes
$T_a=\{1,2\}$ and $T_b=\{3,4\}$.  Top: $z$ read as the sequence of pairs
$(c_t,z_t)$, one vote each, with $\sigma(c_t,z_t)=+1$ when $z_t$ lies in the
recovered target green list $G_K^{m'}(c_t)$ and $-1$ otherwise.  The context
$c_1$ occurs three times, at $t=1,4,9$.  Bottom left: the pairs grouped by
context with the vote sum of each context.  Bottom middle: the attacker knows
$J_K(c_1)\in\{1,2\}$ but not whether it is $1$ or $2$, so $\Pi_O$ draws the
position of every context uniformly from its class, independently across
contexts.  Each draw gives the margins $M_j(J)$, and $E$ holds when every
margin is positive.  $E$ fails when $c_3$ cancels $c_2$ (draw 2) or when no
pair lands on position $2$ (draw 3).  Bottom right: the share of draws with
$E$ is the certificate, a lower bound on the probability that the decoder
reads $z$ as $m'$.  In this small example $\Pr[E]$ is computed by hand,
$2/8$ for $T_a$ (positions of $c_1$ and $c_2$ differ and $c_3$ joins $c_1$)
times $6/8$ for $T_b$ (not all three contexts on one position), whereas the
experiments estimate it by Monte Carlo over draws of $J$.}
\label{fig:certificate}
\end{figure*}

\subsection{Certification}
\label{sec:certification}

\paragraph{Upper bound on attack success.}
The comparisons in Section~\ref{sec:synthesis} reveal the twin class $T(c)$
of the hidden position $J_K(c)$, but cannot distinguish the positions inside
that class.  Recall that
$r_j=(m'_j+m_j^{(1)})\bmod 2$ records whether the target differs from the
reference account at position $j$.  For each twin class $T$, define
\begin{equation}
 a_T:=\sum_{j\in T}r_j.
 \label{eq:class-differences}
\end{equation}
Thus $a_T$ is the number of positions in $T$ at which the target differs from
the reference.  The attacker knows this number, but the observations do not
identify which $a_T$ of the $|T|$ hidden positions carry those differences.
There are $\binom{|T|}{a_T}$ possible placements.  The choices combine across
twin classes, giving the \emph{ambiguity multiplicity}
\begin{equation}
 N(M,m'):=\prod_{T\in\mathcal T}\binom{|T|}{a_T}.
 \label{eq:ambiguity-count}
\end{equation}
This count uses only the known account messages and the chosen target; it does
not require the secret key.

\begin{theorem}[Targeting bound]
\label{thm:targeting-bound}
Under the symmetry assumptions of Appendix~\ref{app:orbit-proof}, every
attack $A$ against the distortion watermarking satisfies, for every target $m'$,
\begin{equation}
 \Pr\!\left[D_K(A(O))=m'\mid O\right]
 \leq\left(\prod_{T\in\mathcal T}\binom{|T|}{a_T}\right)^{-1}.
 \label{eq:targeting-bound}
\end{equation}
\end{theorem}
The proof, including adaptive observations, attack randomness, and decoder
ties, is given in Appendices~\ref{app:orbit-proof} and
\ref{app:closed-orbit-proof}.

\begin{remark}
    In Appendix~\ref{app:orbit-proof}  we present an elegant group-theoretic characterization of the twin sets and relate the group's orbit to the ambiguity multiplicity.
\end{remark}

\begin{figure*}[t]
\centering
\includegraphics[width=.89\textwidth]{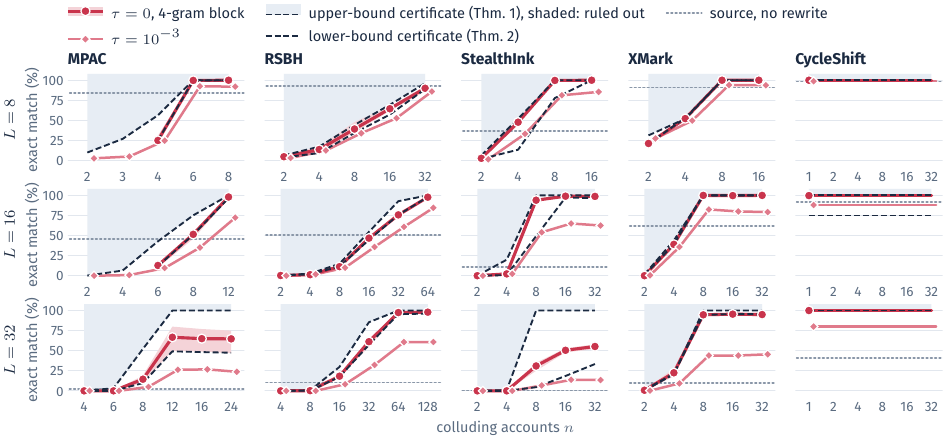}
\caption{Exact-match success against the number of colluding accounts for
five schemes and message lengths $L\in\{8,16,32\}$ at 200 tokens per text,
mean over the four generators Llama-2-7B, Llama-3.1-8B, Qwen3.5-4B, and
Qwen3.5-9B (per-generator panels in
Appendix~\ref{app:collusion}, 64-bit messages and other text lengths in
Appendix~\ref{app:tlen}).
Dashed: the upper-bound certificate, the mean scheme-specific symmetry bound
over the drawn targets (shaded: ruled out),
and the lower-bound certificate, the mean of \eqref{eq:posterior-certificate}
over the $\tau=0$ rewrites.  Circles: the $\tau=0$ rule, band: its range over
the generators.  Diamonds: $\tau=10^{-3}$.  Dotted: the decoder's
exact-match rate on the untouched source.  CycleShift falls to a single
account, so its panels show a random target from one account at $n=1$,
extended flat because any larger coalition can use one account, with the
chunked variant of Appendix~\ref{app:setup} at 32 bits.  Its lower-bound certificate is zero
for Llama-2-7B at 16 bits, whose vocabulary is smaller than the $2^{16}$ shifts,
so one account reveals only part of the keyed set.}
\label{fig:collusion}
\end{figure*}

\paragraph{Attacker-side lower bound.}
\label{sec:certificate}
Theorem~\ref{thm:targeting-bound} limits the success of \emph{any} attack.
Next, we will compute certificates that evaluate how good our attach has inserted the target watermark message $m'$ into the rewrite $z$.
The decoder scores the distinct context--token pairs $\mathcal P(z):=\{(z_{t-\ell:t},z_t):1\leq t\leq|z|\}$.
Each context-token pair induces a hidden position $J_K(z_{t-\ell:t})$.
For every $(c,v)\in\mathcal P(z)$, the attacker computes the twin class $T(c)$ of $J_K(c)$, with $T(c):=[L]$ if no class matched.
Then we collect votes for estimating whether the next token $v$ is in the target message green list or not.
\begin{equation}
 \sigma(c,v):=
 \begin{cases}
  +1, & \text{recovery shows } v\in G^{(m')}_K(c),\\
  -1, & \text{otherwise}.
 \end{cases}
\end{equation}
Recall that the exact position of $J_K(c)$ inside $T(c)$ stays hidden.
We therefore define $\Pi_O$ as the distribution that draws an assignment $J$ uniform on $T(c)$, independently across distinct contexts.
We next define a margin $M_j(J)$ that counts how often we have sampled from the green list as follows:
\begin{equation}
  M_j(J):=\textstyle\sum_{(c,v)\in\mathcal P(z):\,J(c)=j}\sigma(c,v).\\
\end{equation}
Finally, $E$ contains all realization of J for which we have a positive margin for every message position:
\begin{equation}
  E:=\bigl\{J:\ M_j(J)>0\ \text{for every }j\in[L]\bigr\}.
\end{equation}
\begin{theorem}[Certificate validity]
\label{thm:certificate-validity}
Under the posterior $\Pi_O$ and any decoder $D_K$
\begin{equation}
 \begin{aligned}
  \operatorname{Cert}(z,m'\mid O)
  &:=\Pr_{J\sim\Pi_O}[E]\\
  &\leq \Pr[D_K(z)=m'\mid O],
 \end{aligned}
 \label{eq:posterior-certificate}
\end{equation}
with equality if $G^{(m')}_K$ was recovered exactly.
\end{theorem}
The proof is given in Appendix~\ref{app:certificate-proof}.  We estimate the
certificate by Monte Carlo over draws of $J$ (Algorithm~\ref{alg:certificate}
in Appendix~\ref{app:pseudocode}), and it is zero when $N(M,m')>1$.  Under the $\tau=0$ rule every $\sigma$ is $+1$, so $M_j(J)$
counts the pairs at $j$ and $E$ only asks that every position be covered.
The fallback and soft rules also need positive margins.
Theorem~\ref{thm:vote-certificate} (Appendix~\ref{app:theory}) adds an
analytic bound in terms of the true vote counts.

\section{Experiments}
\label{sec:experiments}

\begin{table*}[t]
\centering\footnotesize\setlength{\tabcolsep}{3.5pt}
\resizebox{\ifdim\width>\textwidth\textwidth\else\width\fi}{!}{
\begin{tabular}{@{}llrrrrrr@{\hspace{3\tabcolsep}}rrrrrr@{}}
\toprule
 & & \multicolumn{6}{c}{decoding accuracy (\%) $\uparrow$} & \multicolumn{6}{c}{fidelity to the source $\uparrow$} \\
\cmidrule(lr){3-8}\cmidrule(lr){9-14}
 & & clean & no attack & \multicolumn{2}{c}{own acc.\ $\to\bar m$} & \multicolumn{2}{c}{$n$ accounts} & \multicolumn{2}{c}{BERTScore} & \multicolumn{2}{c}{ROUGE-L} & \multicolumn{2}{c}{LLM judge (1--5)} \\
\cmidrule(lr){5-6}\cmidrule(lr){7-8}\cmidrule(lr){9-10}\cmidrule(lr){11-12}\cmidrule(lr){13-14}
Generator & $L$ & BA & BA & BA & EM & BA & EM & para. & attack & para. & attack & para. & attack \\
\midrule
\multirow{2}{*}{Llama-2-7B} & 8 & 96.5 & 50.2 & 100.0 & 100 & 100.0 & 100 & 0.89 & 0.84 & 0.48 & 0.32 & 3.4 & 2.7 \\
 & 16 & 93.5 & 49.2 & 100.0 & 100 & 99.9 & 99 & 0.90 & 0.84 & 0.48 & 0.32 & 3.4 & 2.9 \\
\multirow{2}{*}{Llama-3.1-8B} & 8 & 98.0 & 49.2 & 100.0 & 100 & 100.0 & 100 & 0.90 & 0.86 & 0.39 & 0.28 & 3.6 & 3.1 \\
 & 16 & 94.8 & 49.4 & 100.0 & 100 & 99.9 & 99 & 0.90 & 0.87 & 0.40 & 0.29 & 3.5 & 3.2 \\
\multirow{2}{*}{Qwen3.5-4B} & 8 & 99.0 & 48.0 & 100.0 & 100 & 100.0 & 100 & 0.90 & 0.86 & 0.43 & 0.30 & 3.2 & 2.5 \\
 & 16 & 93.6 & 50.2 & 99.7 & 96 & 99.6 & 96 & 0.90 & 0.86 & 0.42 & 0.29 & 3.1 & 2.3 \\
\multirow{2}{*}{Qwen3.5-9B} & 8 & 98.4 & 52.5 & 100.0 & 100 & 100.0 & 100 & 0.91 & 0.86 & 0.46 & 0.29 & 3.5 & 2.7 \\
 & 16 & 95.8 & 50.0 & 99.8 & 97 & 99.9 & 98 & 0.91 & 0.86 & 0.45 & 0.29 & 3.3 & 2.5 \\
\bottomrule
\end{tabular}
}
\caption{
Steering binary MPAC on four generators.  BA: bit accuracy of the decoded
message to the attacker's target, EM: exact-match rate, both in percent.
Every cell averages the 100 watermarked texts of one generator and message
length, each rewritten once per condition (800 sources and 2{,}400 rewrites
in all).  $\uparrow$: higher is better.
\underline{Clean}: BA of the watermarked source to its own message.
\underline{No attack}: the same paraphraser, with the
same 4-gram block but no watermark knowledge, scored against a random target,
i.e., chance.  
\underline{Own account}: the attacker uses only its own account and targets the complement of its own message $\bar m$.
\underline{$n$ accounts}: $n=8$
($L=8$) or $n=12$ ($L=16$) colluding accounts and a uniformly random target.
Fidelity to the source, for the paraphrase without the attack (para.) and
for the $n$-account rewrite (attack): BERTScore F1 and ROUGE-L F1 against the watermarked
source, and the mean score of two LLM judges on a 1 to 5 scale
(Appendix~\ref{app:setup}).  
}
\label{tab:models}
\end{table*}

\paragraph{Setup.}
We attack five published multi-bit schemes: \textbf{MPAC}
\citep{yoo-etal-2024-advancing}, \textbf{RSBH} \citep{qu-etal-2025-provably},
\textbf{StealthInk} \citep{jiang-etal-2025-stealthink}, \textbf{XMark}
\citep{xu-etal-2026-xmark}, and \textbf{CycleShift} \citep{fernandez2023three}.  The
watermarked texts are completions of prompts from the RealNewsLike subset of
C4 \citep{raffel2020exploring}, written by four generators, Llama-2-7B
\citep{touvron2023llama}, Llama-3.1-8B \citep{grattafiori2024llama},
Qwen3.5-4B, and Qwen3.5-9B \citep{qwen2026qwen35}.
The attacker uses the same model as chat model for rewriting.
We measure attack efficacy by the exact-match rate, the fraction of rewrites whose decoded message equals the target, and by bit accuracy.
We compare both metrics with the symmetry upper bound of Theorem~\ref{thm:targeting-bound} in the
binary position--color case, with its scheme-specific counterpart otherwise, and
with the attacker's certificate \eqref{eq:posterior-certificate}.
Appendix~\ref{app:setup} gives the corpora, decoding settings, attack
conditions, the formal definitions of every metric, and the
implementation details.

\paragraph{Steering MPAC across generators.}
Table~\ref{tab:models} steers MPAC on all four generators, from the
attacker's own account and from enough colluding accounts to separate every
message position.  Both reach the target in almost every text, for 8- and
16-bit messages alike, while a paraphrase without watermark knowledge stays
at chance.  The rewrites keep the source's topic and part of its wording but
drift from its content.  Targets
that the theory declares unreachable are never decoded exactly
(Figure~\ref{fig:collusion}).  Appendix~\ref{app:schemes} reports the same
table for the other four schemes (Table~\ref{tab:schemes}).

\newcommand{\certcounts}{9{,}200 to 12{,}900 rewrites per scheme and rule}

\begin{table}[!b]
\centering\footnotesize\setlength{\tabcolsep}{4pt}
\begin{tabular}{@{}lrrrrr@{}}
\toprule
 & \multicolumn{4}{c}{certificate} & \\
\cmidrule(lr){2-5}
 & $=1$ & \multicolumn{2}{c}{$\in(0,1)$} & $=0$ & \\
\cmidrule(lr){2-2}\cmidrule(lr){3-4}\cmidrule(lr){5-5}
Scheme & succ. $\uparrow$ & cert. & succ. & succ. & flagged $\downarrow$ \\
\midrule
\multicolumn{6}{@{}l}{\textit{rule } $\tau=0$} \\
\quad MPAC & 100 & 70 & 77 & 4 & 0/104 \\
\quad RSBH & 100 & 8 & 14 & 1 & 6/92 \\
\quad StealthInk & 100 & 60 & 93 & 5 & 0/108 \\
\quad XMark & 100 & 40 & 39 & 0 & 13/108 \\
\addlinespace
\multicolumn{6}{@{}l}{\textit{rule } $\tau=10^{-3}$} \\
\quad MPAC & 100 & 63 & 68 & 2 & 0/129 \\
\quad RSBH & 100 & 12 & 19 & 3 & 3/92 \\
\quad StealthInk & 100 & 43 & 69 & 2 & 2/108 \\
\quad XMark & 100 & 35 & 33 & 0 & 13/108 \\
\bottomrule
\end{tabular}

\caption{Certificate against realized success over all rewrites at 200
tokens under the $\tau=0$ and $\tau=10^{-3}$ rules (\certcounts), grouped
by certificate value: the success rate where the certificate is one, the
mean certificate and the success rate where it lies strictly between zero
and one, and the success rate where it is zero, in percent.  Flagged:
conditions (corpus and account setting) whose mean certificate exceeds the
realized success by more than two standard errors.  $\uparrow$
($\downarrow$): higher (lower) is better.}
\label{tab:cert}
\end{table}

\paragraph{Collusion across schemes.}
Figure~\ref{fig:collusion} sweeps the number of colluding accounts for every
scheme at 8, 16, and 32 bits.  Success tracks the symmetry upper bound and
meets it once the accounts make the target structurally reachable, so each
scheme's residual ambiguity decides how many accounts it withstands:
CycleShift falls to a
single account, MPAC and StealthInk once the accounts separate every message
position or show every symbol value, XMark once every symbol value has
appeared at every block, and RSBH only to coalitions large enough to cover
its codeword symbols.  Message length changes none of this up to 16 bits.
At 32 bits the bound still reaches one, but 200 tokens spread too few votes
over each position for a reliable majority, so success falls short of the
bound until the text grows (Appendix~\ref{app:tlen}), and 64-bit messages
exceed what 200 tokens carry for every scheme.  Across schemes, a paraphrase
without watermark knowledge stays at chance.

\paragraph{Empirical check of the certificate.}
Table~\ref{tab:cert} compares the certificate of every rewrite with whether
the decoder read the target, under the forced rule and under the rule that
leaves low-mass contexts to the paraphraser.  A certificate of one is never
wrong, a certificate of zero almost never sees the target decoded, and in
between the certificate stays below the realized success for MPAC, RSBH,
and StealthInk and matches it for XMark, so the attacker can treat it as
the lower bound it claims to be and release only the rewrites it vouches
for.  The few conditions in which the mean certificate exceeds the realized
success by more than its sampling error concentrate on XMark and RSBH,
whose public positions leave the posterior only the unrevealed colors to
draw.

\section{Conclusion}
Targeted displacement is, to our knowledge, the first attack that steers a multi-bit LLM watermark to a chosen message without the key.
On five published schemes its success meets the predicted upper bound once enough accounts collude.
We argue that our attack can be successfully executed for insiders on closed source LLMs, with potential security implications, including identity theft.

\clearpage
\captionsetup{skip=10pt}
\setlength{\textfloatsep}{20pt plus 2pt minus 4pt}
\setlength{\dbltextfloatsep}{20pt plus 2pt minus 4pt}
\setlength{\floatsep}{12pt plus 2pt minus 2pt}
\setlength{\dblfloatsep}{12pt plus 2pt minus 2pt}
\setlength{\abovedisplayskip}{11pt plus 3pt minus 6pt}
\setlength{\belowdisplayskip}{11pt plus 3pt minus 6pt}
\setlength{\abovedisplayshortskip}{0pt plus 3pt}
\setlength{\belowdisplayshortskip}{6.5pt plus 3.5pt minus 3pt}
\renewcommand{\topfraction}{.7}
\renewcommand{\dbltopfraction}{.7}
\renewcommand{\textfraction}{.2}
\setcounter{topnumber}{2}
\setcounter{bottomnumber}{1}
\setcounter{totalnumber}{3}
\setcounter{dbltopnumber}{2}
\section*{Limitations}
The limitations of our paper are the assumptions of its threat model.  Our
attacker knows the watermarking scheme and its public hyperparameters but
not the secret key.  It has the unwatermarked base model and its tokenizer.
It controls accounts whose embedded identifiers it knows, all watermarked
under the provider's single key, and it reads the next-token distributions
of these accounts at the contexts of its rewrite.  These assumptions are,
however, standard in the watermarking literature.  That the attacker knows
everything except the key is Kerckhoffs' principle, which the watermark
stealing threat model adopts explicitly \citep{jovanovic2024watermark}.
Access to the unwatermarked distribution of the model is the available base
responses setting of the same threat model and of its successors
\citep{jovanovic2024watermark,chen2025demark}, and it is the attacker
against which stealthiness is argued \citep{jiang-etal-2025-stealthink}.
Accounts with known identifiers are the spoofing adversary against which
RSBH is argued secure \citep{qu-etal-2025-provably}, and a single provider
key with the identifier as the message is how every scheme we study embeds
it
\citep{fernandez2023three,yoo-etal-2024-advancing,qu-etal-2025-provably,jiang-etal-2025-stealthink,xu-etal-2026-xmark}.
Next-token log-probabilities are returned by public APIs and assumed by
prior attacks \citep{pang2024no,chen2025demark}.  Our attacker reads the
full distribution, which idealizes the top-$k$ entries that APIs return.
A provider that hides account identifiers, keys the labeling per account,
or withholds log-probabilities removes what the attacker reads, and we do
not measure the query cost of sample access, under which the symmetry bound
still holds.

\section*{Ethics Statement}
This work was conducted in accordance with the ACL Code of Ethics and the
ACM Code of Ethics and Professional Conduct, and we have adhered to the
rules throughout this research.  We recognize that the attack we describe
could be misused to make a watermark decoder attribute a harmful text to a
user who did not write it, and that decoded identifiers may feed moderation
or legal decisions.  Our intention is to show providers and scheme
designers that attribution by a multi-bit watermark can be displaced
without the key, so that a decoded identifier is not treated as conclusive
evidence on its own and so that future schemes are designed against this
threat.  The symmetry bound also tells a provider which targets no attacker can
reach.  All experiments use public prompts, open-weight models, and
accounts that we simulate locally, and no deployed service or real user was
involved.  We also acknowledge that appropriate safeguards and further
evaluation are necessary to mitigate risks in real world applications.

\bibliography{references}

\appendix

\section{Explicit transformations and recovery rules}
\label{app:transformations}

This appendix instantiates the transformation $\Phi$ for every watermark
studied in this paper.  Fix a context $c$ and abbreviate
$p_v=p(v\mid c)$, $q_{i,v}=q_i(v\mid c)$, and
$\ell_i(v)=\log q_{i,v}-\log p_v$.  The equations below describe population
distributions.  Under sample access, $q_i$ is replaced by its empirical
estimate.

\paragraph{General inversion.}
Given an arbitrary function $\Phi$, recall that (P1) states
\begin{equation}
 p_i(\cdot\mid c)
 =\Phi\bigl(p(\cdot\mid c),G_i(c)\bigr).
 \label{eq:set-recovery}
\end{equation}
For some $\Phi$ inverting through it to compute $G_i(c)$ given $p$ and $p_i$ might not be possible.
However, $\Phi$ defines which candidate sets are consistent with the observed distribution.
Under exact distribution access, these candidates are
\begin{equation}
 \begin{aligned}
 \mathcal S_i(c)
 &:=
 \left\{S\subseteq\mathcal V:
 \Phi\bigl(p(\cdot\mid c),S\bigr)=q_i(\cdot\mid c)\right\}\\
 &=
 \underset{S\subseteq\mathcal V}{\arg\min}\;
 d_{\mathrm{TV}}\!\left(
 q_i(\cdot\mid c),
 \Phi\bigl(p(\cdot\mid c),S\bigr)
 \right),
 \end{aligned}
 \label{eq:recover-set}
\end{equation}
where $d_{\mathrm{TV}}(r,s):=\tfrac12\sum_{v\in\mathcal V}|r(v)-s(v)|$ is
total-variation distance.  This inverts (P1): each scheme uses its own public
$\Phi$ (Appendix~\ref{app:transformations}).  If $\mathcal S_i(c)$ is a singleton,
it identifies $G_i(c)$, and otherwise we constrain only memberships on which
all candidates agree, replacing $q_i$ by its empirical estimate under sample
access.  Figure~\ref{fig:recover-set} works through the MPAC case: the
log-ratio of the two distributions takes only two values, and the midpoint
threshold between them returns the set exactly.  By (P2), these memberships
are the only way the key enters decoding, so controlling them suffices to
control the decoder without recovering $K$.

\paragraph{A shared additive transformation.}
MPAC, CycleShift, RSBH, and XMark use the same operation once their
message-conditioned set has been constructed.  For $S\subseteq\mathcal V$ and
bias $\delta>0$, define
\begin{equation}
 \begin{aligned}
  \mathsf A_\delta(p,S)_v
  &:=
  \frac{p_v\exp\!\left(\delta\mathbf 1\{v\in S\}\right)}
       {Z_\delta(p,S)},\\
  Z_\delta(p,S)
  &:=1+\bigl(e^\delta-1\bigr)\sum_{u\in S}p_u.
 \end{aligned}
 \label{eq:app-additive}
\end{equation}
If $q_i=\mathsf A_\delta(p,G_i)$, write
$Z_i:=Z_\delta(p,G_i)$.  Then
\begin{equation}
 \ell_i(v)=
 \begin{cases}
  \delta-\log Z_i,&v\in G_i,\\
  -\log Z_i,&v\notin G_i.
 \end{cases}
 \label{eq:app-additive-levels}
\end{equation}
Thus the exact inversion rule is
\begin{equation}
 \begin{aligned}
 \theta_i&:=\tfrac12\left(\max_v\ell_i(v)+\min_v\ell_i(v)\right),\\
 G_i(c)&=\{v\in\mathcal V:\ell_i(v)>\theta_i\}.
 \end{aligned}
 \label{eq:app-additive-inverse}
\end{equation}
This rule uses only the two observable distributions.  The normalizer is common
to every token and therefore need not be known.

\paragraph{MPAC.}
At context $c$, the key selects a position $J_K(c)$ and a vocabulary coloring
$C_K(c,\cdot)$.  Writing
$V_{K,c,s}:=\{v:C_K(c,v)=s\}$, the message-conditioned set is
\begin{equation}
 G_K^m(c)=V_{K,c,m_{J_K(c)}}.
 \label{eq:app-mpac-set}
\end{equation}
Generation uses
\begin{equation}
 \Phi_{\mathrm{MPAC}}\bigl(p,G_K^m(c)\bigr)
 =\mathsf A_\delta\bigl(p,G_K^m(c)\bigr).
 \label{eq:app-mpac-phi}
\end{equation}
Consequently, a controlled account's set $G_i(c)$ is recovered directly by
\eqref{eq:app-additive-inverse}.  The key, selected position, and absolute
color name are not recovered by this operation.

\paragraph{CycleShift.}
Let $R=\max(2^L,|\mathcal V|)$, let $\langle m\rangle\in[2^L]$ be the integer
represented by the message, and let $H_K(c)\subseteq[R]$ be the keyed base set.
The set associated with message $m$ is the cyclic shift
\begin{equation}
 \begin{aligned}
 G_K^m(c)=\{v\in\mathcal V:\;&
 (v+\langle m\rangle)\bmod R\\[-2pt]
 &{}\in H_K(c)\}.
 \end{aligned}
 \label{eq:app-cycleshift-set}
\end{equation}
Its transformation is again additive,
\begin{equation}
 \Phi_{\mathrm{CS}}\bigl(p,G_K^m(c)\bigr)
 =\mathsf A_\delta\bigl(p,G_K^m(c)\bigr),
 \label{eq:app-cycleshift-phi}
\end{equation}
so \eqref{eq:app-additive-inverse} recovers $G_i(c)$.  When $R=|\mathcal V|$,
the controlled message also recovers the base set and every target set:
\begin{equation}
 \begin{aligned}
 H_K(c)
 &=\{(v+\langle m^{(i)}\rangle)\bmod R:v\in G_i(c)\},\\
 G_K^{m'}(c)
 &=\{v:(v+\langle m'\rangle)\bmod R\in H_K(c)\}.
 \end{aligned}
 \label{eq:app-cycleshift-inverse}
\end{equation}
If $R>|\mathcal V|$, the account reveals only the coordinates of the extended
set that are represented in the model vocabulary.

\paragraph{RSBH.}
Let $\operatorname{Enc}(m)=(a_0(m),\ldots,a_{r-1}(m))$ be the public
Reed--Solomon codeword and let $\sigma(c)\in[r]$ be the public segment selected
by the context.  For each context and symbol value $a$, the key determines a
set $H_{K,c,a}\subseteq\mathcal V$.  Hence
\begin{equation}
 \begin{aligned}
 G_K^m(c)&=H_{K,c,a_{\sigma(c)}(m)},\\
 \Phi_{\mathrm{RSBH}}\bigl(p,G_K^m(c)\bigr)
 &=\mathsf A_\delta\bigl(p,G_K^m(c)\bigr).
 \end{aligned}
 \label{eq:app-rsbh-phi}
\end{equation}
Equation~\eqref{eq:app-additive-inverse} therefore recovers, from account $i$,
the set indexed by its codeword symbol $a_{\sigma(c)}(m^{(i)})$.  It does not
reveal sets belonging to other symbol values.  In the released construction,
the set for symbol value $0$ is generated without the secret multiplier and is
therefore public as well.

\paragraph{XMark.}
A message consists of $B=L/2$ symbols $m_1,\ldots,m_B\in\{0,1,2,3\}$ of two
bits each.  The block of a token is public: if $c=(u,w)$ are the two
preceding tokens, the block is $\sigma(c)=(u+w)\bmod B$.  The key seeds two
permutations of the vocabulary at $c$, each cut into four quarters, which
give every token two quarter labels $Q^1_K(c,v),Q^2_K(c,v)\in\{0,1,2,3\}$.
Symbol $s$ boosts the tokens that lie outside quarter $s$ under both
permutations, the \emph{dark set}
\begin{equation}
 D_{K,c,s}:=\{v:\ Q^1_K(c,v)\neq s\ \text{and}\ Q^2_K(c,v)\neq s\},
 \label{eq:app-xmark-dark}
\end{equation}
which holds $9/16$ of the vocabulary.  Hence
\begin{equation}
 \begin{aligned}
 G_K^m(c)&=D_{K,c,m_{\sigma(c)}},\\
 \Phi_{\mathrm{XM}}\bigl(p,G_K^m(c)\bigr)
 &=\mathsf A_\delta\bigl(p,G_K^m(c)\bigr).
 \end{aligned}
 \label{eq:app-xmark-phi}
\end{equation}
The decoder gives every scored token one vote for its quarter under the
first permutation and, if that quarter differs, one for its quarter under
the second, collects the votes per block, and outputs the symbol with the
fewest votes.  A token votes for $s$ exactly when it lies outside
$D_{K,c,s}$, so the vote count of $s$ at a block is the number of its
tokens outside the dark set of $s$, and the decoder reads the key only
through membership in the sets $G_K^m(c)$, as (P2) requires.  The
released code seeds the two permutations with fixed public constants,
which we replace by a secret key and a second key derived from it.
Equation~\eqref{eq:app-additive-inverse} recovers from account $i$ the
dark set $D_{K,c,a}$ of its symbol $a=m^{(i)}_{\sigma(c)}$ and hence its
complement $U_{K,c,a}:=\mathcal V\setminus D_{K,c,a}$, the tokens that
lie in quarter $a$ under at least one permutation.  It reveals neither the
two quarters separately nor the dark sets of the other symbol values.
Because the block is public, the sets of different accounts at the same
context combine by symbol value.  The target's dark set at a context is
available exactly when some account carries the target's symbol $s'$ at
that block, or when at most one symbol value is unseen there: if $a$ and
$b$ are two known values other than $s'$, a token in
$U_{K,c,a}\cap U_{K,c,b}$ has one quarter equal to $a$ and the other equal
to $b$, so both differ from $s'$ and the token lies in $D_{K,c,s'}$.  The
implementation forces the tokens that lie in at least two known non-target
unions, a subset of $D_{K,c,s'}$ that is certain under every labeling
consistent with the observations.  When the target's symbol is unseen and
two or more symbol values are unseen at a block, these values are
interchangeable in the attacker's view, and the orbit of the target gains a
factor equal to their number.

\paragraph{StealthInk.}
Here the exact transformation uses an ordered vocabulary partition, not only
the selected chunk as an unordered set.  Let
$\pi_K(c)=(\pi_1,\ldots,\pi_{|\mathcal V|})$ be the keyed permutation, split
into contiguous chunks.  If $s=m_{J_K(c)}$ and the index interval of chunk $s$
is $\{a+1,\ldots,b\}$, then
\begin{equation}
 G_K^m(c)=\{\pi_{a+1},\ldots,\pi_b\}.
 \label{eq:app-stealthink-set}
\end{equation}
Put $F_0=0$ and
$F_r=\sum_{j=1}^{r}p_{\pi_j}$, and let
$I=[F_a,F_b]$ and $I^\star=[1-F_b,1-F_a]$.  Define
\begin{equation}
 w_I(x):=1+\mathbf 1\{x\in I^\star\}-\mathbf 1\{x\in I\}.
 \label{eq:app-stealthink-weight}
\end{equation}
The released cumulative reweighting is exactly
\begin{equation}
 \Phi_{\mathrm{SI}}(p,\pi_K(c),s)_{\pi_r}
 =
 \int_{F_{r-1}}^{F_r}w_I(x)\,dx.
 \label{eq:app-stealthink-phi}
\end{equation}
The selected part of the probability interval receives weight $0$, its
reflection about $1/2$ receives weight $2$, and their overlap and the
remainder receive weight $1$.  Since a token in the selected chunk has its
whole cumulative interval inside $I$, whereas every other token lies outside
$I$, the exact membership implications are
\begin{equation}
 \begin{aligned}
 q_{i,v}<p_v&\ \Longrightarrow\ v\in G_i(c),\\
 q_{i,v}>p_v&\ \Longrightarrow\ v\notin G_i(c),\\
 q_{i,v}=p_v&\ \Longrightarrow\ v\text{ is unresolved}.
 \end{aligned}
 \label{eq:app-stealthink-inverse}
\end{equation}
Thus this transformation generally reveals only part of $G_i(c)$ and part of
its complement.  The implementation uses the numerically conservative
log-ratio rules $\log(q_{i,v}/p_v)<-5$ for \emph{inside} and
$\log(q_{i,v}/p_v)>0.2$ for \emph{outside}, leaving the remainder unresolved.
When applying \eqref{eq:recover-set} to this scheme, a candidate consists of a
chunk together with a compatible ordering $\pi$, because the ordering affects
the output probabilities.

\section{Pseudocode of the attack on MPAC}
\label{app:pseudocode}

Algorithm~\ref{alg:mpac-attack} lists the attack of Section~\ref{sec:method}
for binary MPAC, with the green-list recovery rule of
Appendix~\ref{app:transformations}, and Algorithm~\ref{alg:certificate} the
certificate of Section~\ref{sec:certification}.  Both are written for one
source text.  The attacker queries its accounts and the base model at a
context once and caches the recovered lists.  With exact next-token
distributions the log-ratio $\ell_i$ takes exactly two values at a context,
so the midpoint threshold recovers $G_i(c)$ exactly, and under sample access
the empirical estimates of Appendix~\ref{app:transformations} take its place.
The experiments use $S=1{,}000$ draws for the certificate.

\begin{algorithm*}[t]
\caption{Targeted displacement attack on binary MPAC}
\label{alg:mpac-attack}
\small
\begin{algorithmic}[1]
\REQUIRE source $y$, target $m'\in\{0,1\}^L$, messages $m^{(1)},\ldots,m^{(n)}$ of the controlled accounts, next-token access to the accounts $q_i(\cdot\mid c)$ and to the base model $p(\cdot\mid c)$, paraphraser $p_{\mathrm{para}}$, threshold $\tau$, bias $\delta_{\mathrm{atk}}$
\ENSURE rewrite $z$
\STATE \COMMENT{twin classes and relative bits, from the messages alone (Section~\ref{sec:synthesis})}
\FOR{$j=1,\ldots,L$}
 \STATE $x_j\gets\bigl((m^{(i)}_j+m^{(1)}_j)\bmod 2\bigr)_{i=2}^{n}$ \quad and \quad $r_j\gets(m'_j+m^{(1)}_j)\bmod 2$
\ENDFOR
\STATE $\mathcal T\gets$ partition of $[L]$ into classes of equal $x_j$
\FOR{every class $T\in\mathcal T$}
 \IF{$r_j$ is constant on $T$}
  \STATE $r_T\gets$ that value
 \ELSE
  \STATE $r_T\gets r_j$ for a uniformly random $j\in T$ \COMMENT{ambiguous class, $N(M,m')>1$, success bounded by Theorem~\ref{thm:targeting-bound}}
 \ENDIF
\ENDFOR
\STATE \COMMENT{constrained rewriting (Section~\ref{sec:constructive-attack})}
\STATE $z\gets()$
\FOR{$t=1,2,\ldots$ until the paraphraser stops}
 \STATE $p_t(\cdot)\gets p_{\mathrm{para}}(\cdot\mid y,z_{<t})$, then $p_t(v)\gets0$ for every $v$ that would complete a 4-gram already in $z_{<t}$
 \STATE $c_t\gets$ the watermark context of $z_{<t}$ (the previous token for MPAC), \quad $B_t\gets\textsc{TargetList}(c_t)$
 \IF{$B_t\neq\emptyset$ and $Z_t:=\sum_{u\in B_t}p_t(u)\geq\tau$}
  \STATE sample $z_t\sim r_t$ with $r_t(v)=p_t(v)\,\mathbf 1\{v\in B_t\}/Z_t$ \hfill \eqref{eq:hard-attack}
 \ELSE
  \STATE sample $z_t\sim r_t$ with $r_t(v)\propto p_t(v)\exp\!\bigl(\delta_{\mathrm{atk}}\mathbf 1\{v\in B_t\}\bigr)$ \hfill \eqref{eq:soft-attack}
 \ENDIF
\ENDFOR
\RETURN $z$
\STATE
\STATE \textbf{procedure} \textsc{TargetList}$(c)$ \COMMENT{the green list of $m'$ at $c$, \eqref{eq:target-set}}
\FOR{$i=1,\ldots,n$}
 \STATE $\ell_i(v)\gets\log q_i(v\mid c)-\log p(v\mid c)$ for all $v\in\mathcal V$, \quad $\theta_i\gets\tfrac12\bigl(\max_v\ell_i(v)+\min_v\ell_i(v)\bigr)$, \quad $G_i(c)\gets\{v:\ell_i(v)>\theta_i\}$
\ENDFOR
\STATE $x(c)\gets\bigl(\mathbf 1\{G_i(c)\neq G_1(c)\}\bigr)_{i=2}^{n}$
\IF{$x(c)=x_j$ for some $j\in[L]$}
 \STATE $T(c)\gets$ the class with signature $x(c)$
 \RETURN $G_1(c)$ if $r_{T(c)}=0$, else $\mathcal V\setminus G_1(c)$
\ELSE
 \RETURN $\emptyset$ \COMMENT{no class matches, the paraphraser alone decides this step}
\ENDIF
\end{algorithmic}
\end{algorithm*}

\begin{algorithm}[t]
\caption{Attacker-side certificate}
\label{alg:certificate}
\small
\begin{algorithmic}[1]
\REQUIRE rewrite $z$, the classes $\mathcal T$, relative bits $r_T$, and procedure \textsc{TargetList} of Algorithm~\ref{alg:mpac-attack}, number of draws $S$
\ENSURE $\operatorname{Cert}(z,m'\mid O)$
\IF{$N(M,m')>1$}
 \RETURN $0$
\ENDIF
\STATE $\mathcal P(z)\gets\{(c_t,z_t):1\le t\le|z|\}$ \COMMENT{distinct pairs}
\FOR{every distinct context $c$ in $\mathcal P(z)$}
 \STATE $T(c)\gets$ the class matched in \textsc{TargetList}$(c)$, or $[L]$ if none
 \STATE $\sigma(c,v)\gets+1$ if $v\in\textsc{TargetList}(c)$, else $-1$, for every $(c,v)\in\mathcal P(z)$
\ENDFOR
\STATE $\mathrm{ok}\gets0$
\FOR{$s=1,\ldots,S$}
 \STATE draw $J(c)\sim\mathrm{Unif}\,T(c)$ independently for every distinct context $c$
 \STATE $M_j(J)\gets\sum_{(c,v)\in\mathcal P(z):\,J(c)=j}\sigma(c,v)$ for $j=1,\ldots,L$
 \IF{$M_j(J)>0$ for every $j\in[L]$}
  \STATE $\mathrm{ok}\gets\mathrm{ok}+1$
 \ENDIF
\ENDFOR
\RETURN $\mathrm{ok}/S$
\end{algorithmic}
\end{algorithm}

\section{Experimental details}
\label{app:setup}

\paragraph{Corpora.}
Every corpus consists of 100 completions of C4 prompts
\citep{raffel2020exploring} with 200 tokens each, sampled at temperature 0.7
and top-$p$ 0.9 \citep{holtzman2020curious} from the generator under one
secret key and one attacker message, with binary MPAC ($\gamma=1/2$,
$\delta=2$).  Message lengths are $L\in\{8,16\}$ in the tables and up to
$L=64$ in Figure~\ref{fig:collusion}.  The released CycleShift code
enumerates $2^L$ shifts, so at 32 bits we run a chunked variant: four 8-bit
blocks, the block chosen by the previous token, each a released CycleShift
over the same secret set.

\paragraph{Attacker.}
The attacker rewrites each text with a chat model that shares the
generator's tokenizer (Llama-2-7B-chat, Llama-3.1-8B-Instruct, and the Qwen
model itself with its reasoning mode disabled) and has distribution access
to its accounts.  The paraphraser samples at temperature 0.7 and top-$p$
0.9.  We report the constrained rewriting rule of
Section~\ref{sec:constructive-attack} with $\tau=0$, and the paraphraser
additionally blocks repeated 4-grams, which keeps constrained decoding from
collapsing into loops.  Conditions: the own account with the target set to
the complement $\bar m$ or to a message one bit away, and $n$ colluding
accounts with a uniformly random target.  Fact protection
(used in the fidelity experiments of Appendix~\ref{app:fidelity}) leaves a
step unconstrained when the
paraphraser's top token has probability at least $\kappa$ and ends the text,
is whitespace, or continues a numeral, name, or negation copied from the
source.  Such a token may vote against the target, and the certificate counts
it as a rival vote.

\paragraph{Metrics.}
Let $\hat z$ be the rewrite of a source $y$ against target $m'$.  Success,
or exact match, is the fraction of texts with $D_K(\hat z)=m'$, and bit
accuracy is the fraction of the target's bits (symbols, for the 2-bit
schemes) that the decoded message matches, both over the 100 texts of a
corpus.  The bound is the mean of the appropriate symmetry upper bound over
the drawn targets; in the binary position--color case this is
$1/N(M,m')$ from Theorem~\ref{thm:targeting-bound}.  The certificate
is the Monte Carlo estimate of the posterior quantity
\eqref{eq:posterior-certificate} (Table~\ref{tab:cert}).  Fidelity to the
source uses BERTScore F1 \citep{zhang2020bertscore} and ROUGE-L F1
\citep{lin-2004-rouge} between $\hat z$ and $y$.  Two LLM judges
\citep{zheng2023judging} from different families, Llama-3.3-70B-Instruct
\citep{grattafiori2024llama} and Qwen2.5-72B-Instruct \citep{qwen2024qwen25},
both in 4-bit AWQ \citep{lin2024awq} with greedy decoding, rate on a
five-point scale how faithfully $\hat z$ keeps the facts, numbers, names,
negations, and claims of $y$ (5 for the same content, 3 when at least one
fact is changed, invented, or lost, and 1 for different content).  A
DeBERTa-v3-large NLI model \citep{he2023debertav3} gives \emph{support}, the
mean over the sentences of $\hat z$ of the probability that $y$ entails the
sentence, and \emph{coverage}, the mean over the sentences of $y$ of the
probability that $\hat z$ entails the sentence.  \emph{Facts kept} is the
fraction of the numeric expressions of $y$, with sign, currency, and
magnitude word, and of its named entities, found by a BERT NER tagger
trained on CoNLL-2003
\citep{devlin-etal-2019-bert,tjong-kim-sang-de-meulder-2003-introduction},
that $\hat z$ reproduces verbatim (numbers up to spacing and equivalent
percent and dollar notation).  Neither extractor shares code with fact
protection.  The judges agree with each other (Spearman 0.78 over the
rewrites and paraphrases, within one point on 97 percent of the pairs) and
with NLI (0.64).  For the other schemes, Table~\ref{tab:schemes} reports
BERTScore and perplexity under Llama-2-13B instead, since ROUGE-L and the
judges were computed for the MPAC rewrites.

\paragraph{Implementation details.}
\label{app:impl}
Every generation, attack, and scoring job runs on one NVIDIA H100 (94\,GB)
or A40 (48\,GB) of a Slurm cluster with 8 CPU cores and 64\,GB of host
memory (96\,GB for the judges); the LLM judges run on two A40s with tensor
parallelism and the NLI and NER scorers on one A40 (wall-clock costs in
Appendix~\ref{app:runtime}).  The code is Python~3.12 with PyTorch~2.11
\citep{paszke2019pytorch} and Transformers~5.10
\citep{wolf-etal-2020-transformers}, except that the earliest Llama-2-7B
MPAC runs at 8 and 16 bits with $\tau=10^{-3}$ used PyTorch~2.4.1 and
Transformers~4.45.2; the judges are served by vLLM~0.23
\citep{kwon2023efficient}.  The generators are the Hugging Face checkpoints
\texttt{Llama-2-7b-hf}, \texttt{Llama-3.1-8B}, \texttt{Qwen3.5-4B}, and
\texttt{Qwen3.5-9B}, and the paraphrasers \texttt{Llama-2-7b-chat-hf},
\texttt{Llama-3.1-8B-Instruct}, and the two Qwen models themselves, so one
copy of a Qwen model serves as generator, attacker's base model, and
paraphraser.  Llama-2-7B generates in float16 and the other generators in
bfloat16; during the attack every model runs in bfloat16, the Llama
paraphrasers use eager attention, a guard replaces non-finite logits during
rewriting (Llama-2-7B-chat produced NaNs in batched float16 generation),
and the learner compares log-probabilities in float32.

Corpora are sampled in batches of 10 prompts, truncated to 1{,}200 tokens
and left-padded, with exactly $T$ new tokens, top-$k$ disabled, and a
5-gram repetition block \citep{paulus2018deep}; the unwatermarked
completion of a prompt reuses the sampling seed of its watermarked twin.
Rewrites are sampled in batches of 8 with top-$k$ disabled.  The 4-gram
block of Tables~\ref{tab:models} and \ref{tab:schemes} and of
Figure~\ref{fig:collusion} is the \texttt{no\_repeat\_ngram\_size} option
of Transformers, which indexes the prompt as well, so a rewrite repeats no
4-gram of its source; the rewrite-only block of Appendix~\ref{app:fidelity}
is our own logits processor over the generated tokens.  The watermarked API
is simulated in-process: an oracle takes the generator's logits at the
attacker's prefix, applies the watermark of the queried account (the
$\delta$ bias of its set, or StealthInk's reweighting), and returns the full
log-softmax; the base-model forward pass is shared with the attacker's own
call and cached per prefix.  For the additive schemes the learner
classifies every token of the vocabulary by thresholding the log-ratio
$\log p_i(v\mid c)-\log p(v\mid c)$, which equals
$\delta\,\mathbf 1\{v\in G_i(c)\}-\log Z$ up to float32 rounding, at the
midpoint of its range; StealthInk uses the closed-form rule of
Appendix~\ref{app:transformations}.  Learned sets are cached per context
and account, and for MPAC accounts whose learned sets agree on at least 90
percent of the jointly resolved tokens form one twin class.  Each corpus
carries a pool of extra accounts with random messages (16 for MPAC at 8 and
16 bits, up to 128 for the longer messages and for the schemes with more
positions); the $n$ colluding accounts are the attacker's own account and
the first $n-1$ accounts of one fixed random ordering of the pool, so the
coalitions are nested in $n$, and the target of a text is drawn by a
generator seeded with the condition, the global seed, and the text index,
so a text faces the same target under every rewrite rule.  The certificate
\eqref{eq:posterior-certificate} uses 1{,}000 Monte Carlo draws of the
position assignment for MPAC, 300 for RSBH and StealthInk, and 200 for
XMark, and is exact for CycleShift.  Scheme parameters follow the released
code: $\delta=2$ wherever a logit bias is added (StealthInk reweights
without a bias), one-token contexts except for XMark (two tokens) and
StealthInk (three), and set fractions $\gamma=1/2$ for MPAC and RSBH,
$1/4$ for CycleShift, and $9/16$ for XMark.  Our MPAC re-implementation
draws the position and the color partition of a context from a PRNG
seeded by a hash of key and context rather than from
\texttt{torch.randperm}, so that generation and detection agree across
devices; an 8-bit corpus generated with the released MPAC code and judged
by its detector gives the same success rates as our re-implementation.
The other schemes keep the seeding and decoding of their released code
(Appendix~\ref{app:transformations}), with a secret key where the released
code uses a fixed constant.

BERTScore F1 uses \texttt{bert-score}~0.3.13 with RoBERTa-large and no
baseline rescaling, ROUGE-L F1 uses \texttt{rouge-score}~0.1.2 with Porter
stemming, and perplexity is that of Llama-2-13B in float16 given the last
3{,}000 characters of the prompt.  The judges,
\texttt{Llama-3.3-70B-Instruct-AWQ} and \texttt{Qwen2.5-72B-Instruct-AWQ},
generate at most 400 tokens in a 4{,}096-token context and answer with a
JSON object holding a score and up to three changed facts; the score is
read from the reply with a regular expression, and a reply without a score
is dropped.  The NLI model is
\texttt{DeBERTa-v3-large-\allowbreak{}mnli-fever-\allowbreak{}anli-ling-wanli} in float16 on pairs
truncated to 512 tokens, with sentences split at terminal punctuation and
kept when they have at least four words, and the NER tagger is
\texttt{bert-large-NER} with simple entity aggregation.  All randomness is
seeded (seed 0 for the secret key, the messages, sampling, the attack, and
the choice of accounts), and every reported cell is a single run over the
100 texts of its corpus.

\section{Paraphraser prompt}
\label{app:prompt}

Every rewrite in the paper, the plain paraphrase included, is generated from
the same two chat messages, and the rules of Section~\ref{sec:constructive-attack}
change only the decoding.  The source $y$ is the watermarked completion alone,
without the C4 prompt that produced it.

\begin{center}\small\setlength{\fboxsep}{5pt}\microtypesetup{protrusion=false}
\fbox{\parbox{\dimexpr\linewidth-2\fboxsep-2\fboxrule\relax}{\raggedright
\textbf{System}\par
\texttt{You are a careful editor.}\par\medskip
\textbf{User}\par
\texttt{Rewrite the following passage in your own words. Preserve every fact
and the order of ideas, keep roughly the same length, and output only the
rewritten passage with no preamble.}\par\medskip
\texttt{Passage:}\par
\textit{$\langle$watermarked text $y\rangle$}}}
\end{center}

Each paraphraser receives the two messages in its own chat template and writes
the assistant turn.  The template of Llama-3.1-8B-Instruct adds its default
lines \texttt{Cutting Knowledge Date: December 2023} and \texttt{Today Date:
26 Jul 2024} to the system message, and for Qwen3.5 the reasoning mode is
disabled, so the prompt ends with an empty reasoning block.  The answer is
sampled at temperature 0.7 and top-$p$ 0.9 with 150 to 248 new tokens (the
sources have 200), and the decoder scores it as generated, without
post-processing.

\section{Example rewrites}
\label{app:examples}

Table~\ref{tab:examples} shows two watermarked sources with their plain
paraphrase and their attack rewrite under the $\tau=0$ rule of
Table~\ref{tab:models}, both decoded to the target.  Selection: among the
792 $n$-account rewrites that reached their target, we keep the 75 that both
LLM judges score at least 4 (facts kept), leave out the Llama-2-7B-chat
paraphraser, whose $\tau=0$ rewrites carry spelling and case corruption that
the judges tolerate but perplexity (Table~\ref{tab:schemes}) does not, and
show the two closest to the source by ROUGE-L among the rest, at most one
per generator.  Over all 792 the ROUGE-L median is 0.29, and the examples score
0.40 and 0.51.  The texts are shown as the judges and the metrics saw them,
after the paraphraser's chat preamble is removed.

\begin{table*}[p]
\centering\footnotesize\setlength{\tabcolsep}{4pt}\microtypesetup{protrusion=false}
\begin{tabular}{@{}>{\raggedright\arraybackslash}p{0.315\textwidth}>{\raggedright\arraybackslash}p{0.315\textwidth}>{\raggedright\arraybackslash}p{0.315\textwidth}@{}}
\toprule
Watermarked source & Paraphrase without the attack & Attack rewrite \\
\midrule
\multicolumn{3}{@{}p{\dimexpr0.945\textwidth+4\tabcolsep\relax}@{}}{\textit{Example 1.} Qwen3.5-4B, $L=16$, 12 colluding accounts.  Target $1001000100100101$, decoded $1001000100100101$.  Paraphrase / attack rewrite: judge 5.0 / 4.0, ROUGE-L 0.53 / 0.51, BERTScore 0.93 / 0.93.} \\
\addlinespace[3pt]
Google Music will continue to be free to download in the U.S., but the service is shutting down for free users in all of Europe, except for Italy. Apple Music, now in beta testing, seems to be getting the best of all worlds. It offers a monthly subscription service like Spotify and Google Music, but the company is also offering an option to buy your music from the App Store. Apple Music is available for testing in the U.S. only for now and it will be available as an in-app purchase on both iOS and Mac for users in beta. So there are a lot of choices for mobile music consumers, and Google Music is a strong contender for those looking for a cloud music experience. Does the “cloud music” label mean that it’s only available on mobile devices? No. You can use the Google Music app on your iPhone and Android. You can also access the service through a web browser. How much of your music is stored on your computer & Google Music remains free to download within the U.S.; however, the service is ending for free users across all European countries except Italy. Apple Music, currently in beta, appears to combine the best features of both worlds: it offers a monthly streaming subscription similar to Spotify and Google, while also allowing users to purchase individual tracks from the App store. For now, Apple Music is available exclusively for testing in America, where beta users can access it as an in-purchase on both iOS devices and Mac computers. Consequently, mobile music consumers have numerous options, with Google Music serving as a strong choice for those seeking a cloud-based music experience. The term "cloud music" does not imply the service is limited to mobile devices; users can stream via the app on iPhones and Android phones, as well as through any web browser. & Users will still be able to download Google Music for free from America. However, in every part Europe other except for Italy, free access to music is shutting down.\par Apple Music, which is now in beta, appears to combine the best of both options: a subscription service that works the type of Spotify and Google plus an option to purchase songs from Apple Store.\par Apple is available only for testing right now in America. It is available as an on-the-floor, in-app store purchase to iOS and MAC users who participate in beta testing the service.\par So there are many options for mobile consumers and Google Music remains a very strong contender if people are searching out the best cloud-based music experience.\par The label "cloud music" does not mean it is only for mobile phones or tablets. You can utilize Google Music on an iOS and Android phone, plus on the web.\par So much of your personal library stored on the personal computer. \\
\midrule
\multicolumn{3}{@{}p{\dimexpr0.945\textwidth+4\tabcolsep\relax}@{}}{\textit{Example 2.} Llama-3.1-8B, $L=16$, 12 colluding accounts.  Target $1100000001010101$, decoded $1100000001010101$.  Paraphrase / attack rewrite: judge 4.0 / 4.0, ROUGE-L 0.41 / 0.40, BERTScore 0.91 / 0.89.} \\
\addlinespace[3pt]
They also collected money for those in struggle and leafleted with information. The event went on for two hours. The speakers at the rally included a young man from Afghanistan, who said, “The U.S. has been bombing my country for 16 years. We don’t want this war. We don’t support this war. We support peace, not war.” The speaker from the Philippines, who is also a veteran of the U.S.’ wars there, said, “The people of the Philippines have been resisting the U.S. for more than 120 years. The American imperialists have never been able to subjugate the Filipino people, and they never will.” The IAC also set up an anti-imperialist book table on Wilshire Boulevard at the corner of Western. The table was staffed by IAC members who spoke to passers-by about the U.S. war in Afghanistan, the U.S. wars in Korea and the Philippines, and the U.S. attack on & The event featured fundraising efforts for those affected by conflict, as well as the distribution of informational leaflets. The rally itself lasted for two hours.\par Speakers at the event included a young Afghan man, who emphasized that the U.S.'s 16-year bombing campaign in his country has been met with opposition. He stated that the Afghan people do not support the ongoing war and instead advocate for peace. A speaker from the Phillipines, who has also experienced U.S. military involvement firsthand, noted that the Phillipine people have been resisting U.S. influence for over 120 years, with the U.S.-led imperialist efforts ultimately proving unsuccessful.\par The IAC set up an information table on Wilhire Boulevard at Western Avenue, where members engaged with passersby about U.S. involvement in various conflicts, including Afghanistan, Korea, and the Phillipines. & The demonstration featured fundraising efforts for communities impacted by war. In conjunction, participants distributed leaflets with relevant informational materials during the two-hours-long protest. A speaker at the rally from Afghanistan expressed his country’s prolonged experience of conflict, having been bombed by the U. S for nearly sixteen years. He emphasized that the Afghan people's opposition to war remains unwavering, with their desire for peaceful resolution being the primary stance. A veteran from the U. s-affected Philippines also shared his views at the rally. He highlighted the long history between the U. and s the Philippine people, with their resistance dating back over more  a  century. He confidently stated that the U. could never subsume their people under its control. The IAC set-up a  book display on Wilhire Boulevard, at the intersection of Western. This display was staffed and monitored by IAC volunteers who spoke with on-looking passers-by, addressing the U. 'S involvement in conflicts in countries including Afghanistan, s Korea, s the U. the U. as and s the U, s the attack. \\
\bottomrule
\end{tabular}

\caption{Two example rewrites of binary MPAC under the $\tau=0$ rule of
Table~\ref{tab:models}: the watermarked source, the same paraphraser without
the attack, and the $n$-account attack rewrite that decoded to the target.
Each example's line gives the generator, the message length, the number of
colluding accounts, the target and decoded messages, and the judge score,
ROUGE-L and BERTScore of the paraphrase and of the attack rewrite.
Selection as described in Appendix~\ref{app:examples}.}
\label{tab:examples}
\end{table*}

\section{Results for the other schemes}
\label{app:schemes}

Table~\ref{tab:schemes} gives, for every scheme at 8 bits and for each
generator with the mean over the four, the clean decode rate, the bit
accuracy and exact match of the plain paraphrase, of the own-account rewrite,
and of the $n_{\max}$-account rewrite, and the BERTScore and perplexity of
the rewrites.  The paragraph on other schemes in
Section~\ref{sec:experiments} discusses the means.

\begin{table*}[!tp]
\centering\footnotesize\setlength{\tabcolsep}{3.5pt}
\resizebox{\textwidth}{!}{
\begin{tabular}{@{}lrlrrrrrrrrrr@{}}
\toprule
 & & & clean & no attack & \multicolumn{2}{c}{own acc.\ $\to\bar m$} & \multicolumn{2}{c}{$n_{\max}$ accounts} & \multicolumn{2}{c}{BERTScore $\uparrow$} & \multicolumn{2}{c}{perplexity $\downarrow$} \\
\cmidrule(lr){6-7}\cmidrule(lr){8-9}\cmidrule(lr){10-11}\cmidrule(lr){12-13}
Scheme & $n_{\max}$ & Generator & BA $\uparrow$ & BA $\uparrow$ & BA $\uparrow$ & EM $\uparrow$ & BA $\uparrow$ & EM $\uparrow$ & paraphrase & attack & source & attack \\
\midrule
MPAC (binary) & 8 & Llama-2-7B & 96.5 & 50.2 & 100.0 & 100 & 100.0 & 100 & 0.89 & 0.84 & 4.9 & 41.7 \\
 &  & Llama-3.1-8B & 98.0 & 49.2 & 100.0 & 100 & 100.0 & 100 & 0.90 & 0.86 & 4.4 & 19.8 \\
 &  & Qwen3.5-4B & 99.0 & 48.0 & 100.0 & 100 & 100.0 & 100 & 0.90 & 0.86 & 5.8 & 23.5 \\
 &  & Qwen3.5-9B & 98.4 & 52.5 & 100.0 & 100 & 100.0 & 100 & 0.91 & 0.86 & 5.4 & 24.3 \\
 & & \textit{mean} & 98.0 & 50.0 & 100.0 & 100 & 100.0 & 100 & 0.90 & 0.86 & 5.1 & 27.3 \\
\addlinespace
StealthInk (binary) & 8 & Llama-2-7B & 94.1 & 47.8 & 100.0 & 100 & 100.0 & 100 & 0.89 & 0.82 & 4.2 & 79.4 \\
 &  & Llama-3.1-8B & 97.2 & 50.1 & 100.0 & 100 & 100.0 & 100 & 0.90 & 0.86 & 3.9 & 29.8 \\
 &  & Qwen3.5-4B & 94.9 & 49.6 & 100.0 & 100 & 100.0 & 100 & 0.91 & 0.85 & 5.1 & 33.3 \\
 &  & Qwen3.5-9B & 95.5 & 45.5 & 100.0 & 100 & 100.0 & 100 & 0.91 & 0.85 & 4.7 & 37.5 \\
 & & \textit{mean} & 95.4 & 48.2 & 100.0 & 100 & 100.0 & 100 & 0.90 & 0.84 & 4.5 & 45.0 \\
\addlinespace
StealthInk (2-bit chunks) & 16 & Llama-2-7B & 76.5 & 25.5 & 35.8 & 1 & 100.0 & 100 & 0.90 & 0.85 & 4.0 & 42.0 \\
 &  & Llama-3.1-8B & 82.5 & 24.0 & 36.0 & 1 & 100.0 & 100 & 0.90 & 0.88 & 3.6 & 19.3 \\
 &  & Qwen3.5-4B & 75.5 & 23.0 & 36.0 & 1 & 100.0 & 100 & 0.91 & 0.87 & 4.9 & 24.3 \\
 &  & Qwen3.5-9B & 72.5 & 25.0 & 36.0 & 1 & 100.0 & 100 & 0.91 & 0.87 & 4.5 & 26.8 \\
 & & \textit{mean} & 76.8 & 24.4 & 35.9 & 1 & 100.0 & 100 & 0.90 & 0.87 & 4.3 & 28.1 \\
\addlinespace
XMark (2-bit chunks) & 16 & Llama-2-7B & 97.5 & 28.2 & 22.5 & 0 & 100.0 & 100 & 0.90 & 0.85 & 4.7 & 39.0 \\
 &  & Llama-3.1-8B & 98.2 & 26.2 & 22.0 & 0 & 100.0 & 100 & 0.90 & 0.88 & 4.3 & 17.2 \\
 &  & Qwen3.5-4B & 96.8 & 25.2 & 27.0 & 0 & 100.0 & 100 & 0.90 & 0.88 & 5.8 & 22.2 \\
 &  & Qwen3.5-9B & 96.2 & 25.0 & 23.0 & 0 & 100.0 & 100 & 0.91 & 0.87 & 5.2 & 23.7 \\
 & & \textit{mean} & 97.2 & 26.2 & 23.6 & 0 & 100.0 & 100 & 0.90 & 0.87 & 5.0 & 25.5 \\
\addlinespace
CycleShift & 1 & Llama-2-7B & 99.5 & 52.6 & 100.0 & 100 & 100.0 & 100 & 0.90 & 0.79 & 5.0 & 117.9 \\
 &  & Llama-3.1-8B & 99.5 & 51.0 & 100.0 & 100 & 100.0 & 100 & 0.90 & 0.83 & 4.7 & 41.6 \\
 &  & Qwen3.5-4B & 97.2 & 50.0 & 100.0 & 100 & 100.0 & 100 & 0.90 & 0.82 & 6.1 & 29.1 \\
 &  & Qwen3.5-9B & 100.0 & 51.9 & 100.0 & 100 & 100.0 & 100 & 0.91 & 0.82 & 5.7 & 38.3 \\
 & & \textit{mean} & 99.1 & 51.4 & 100.0 & 100 & 100.0 & 100 & 0.90 & 0.82 & 5.4 & 56.7 \\
\addlinespace
RSBH & 32 & Llama-2-7B & 97.0 & 47.9 & 38.6 & 0 & 97.6 & 93 & 0.89 & 0.85 & 4.7 & 38.3 \\
 &  & Llama-3.1-8B & 97.8 & 49.8 & 45.6 & 0 & 96.6 & 90 & 0.90 & 0.87 & 4.3 & 17.2 \\
 &  & Qwen3.5-4B & 97.2 & 52.2 & 44.8 & 0 & 96.8 & 91 & 0.91 & 0.87 & 5.7 & 20.5 \\
 &  & Qwen3.5-9B & 98.9 & 52.1 & 45.8 & 0 & 95.1 & 86 & 0.91 & 0.87 & 5.3 & 22.2 \\
 & & \textit{mean} & 97.7 & 50.5 & 43.7 & 0 & 96.5 & 90 & 0.90 & 0.86 & 5.0 & 24.5 \\
\bottomrule
\end{tabular}
}
\caption{Steering other multi-bit schemes with 8-bit messages, per generator
with the mean over the four generators.  Same accuracy metrics as
Table~\ref{tab:models}, with BERTScore and perplexity under Llama-2-13B
(source and attack rewrite) as text metrics.  $\uparrow$ ($\downarrow$):
higher (lower) is better.  The colluding column uses the
largest number of accounts run for each scheme, $n_{\max}$, and for
CycleShift a random target from the own account alone.  CycleShift: the
payload cyclically shifts one secret set, so one account reveals every
payload's set.  RSBH: public segments and an RS(4,2) code over GF(16) whose
released decoder falls back to the raw payload when decoding fails.
XMark: public blocks from the two previous tokens and two keyed quarterings
of the vocabulary, with the 2-bit symbol suppressing one quarter under both.
StealthInk: the released 2-bit chunks, for which the accuracy columns report
symbol accuracy over the four 2-bit positions, and a binary variant.}
\label{tab:schemes}
\end{table*}

\section{Per-generator collusion curves}
\label{app:collusion}

Figures~\ref{fig:collusion-mpac} to~\ref{fig:collusion-cycleshift} repeat the
panels of Figure~\ref{fig:collusion} for every generator separately, with the
same rules and bound, and the MPAC figure adds the soft bias of
\eqref{eq:soft-attack}.

\begin{figure*}[t]
\centering
\includegraphics[width=\textwidth]{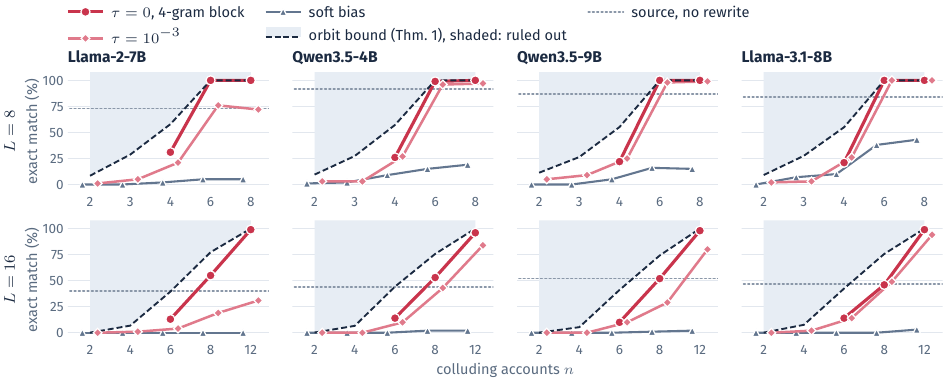}
\caption{Binary MPAC: per-generator version of Figure~\ref{fig:collusion}.
Triangles: the soft bias of \eqref{eq:soft-attack}.}
\label{fig:collusion-mpac}
\end{figure*}

\begin{figure*}[t]
\centering
\includegraphics[width=\textwidth]{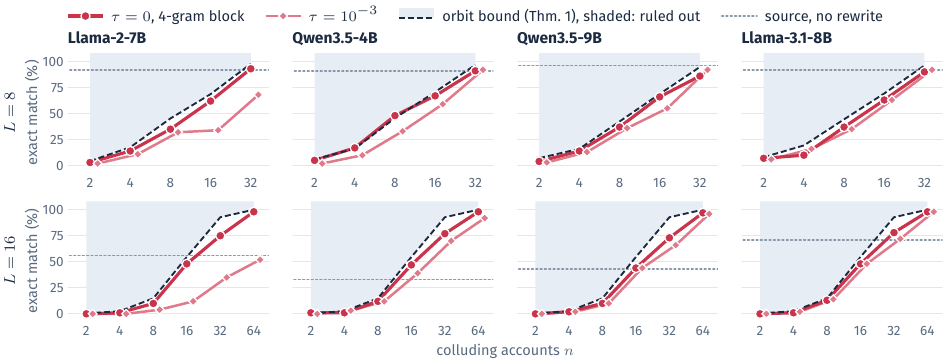}
\caption{RSBH: per-generator version of Figure~\ref{fig:collusion}.  A
target is reachable once at most one of its codeword symbols is uncovered by
the accounts, which takes 32 accounts for 8-bit messages and 64 for 16-bit
ones.}
\label{fig:collusion-rsbh}
\end{figure*}

\begin{figure*}[t]
\centering
\includegraphics[width=\textwidth]{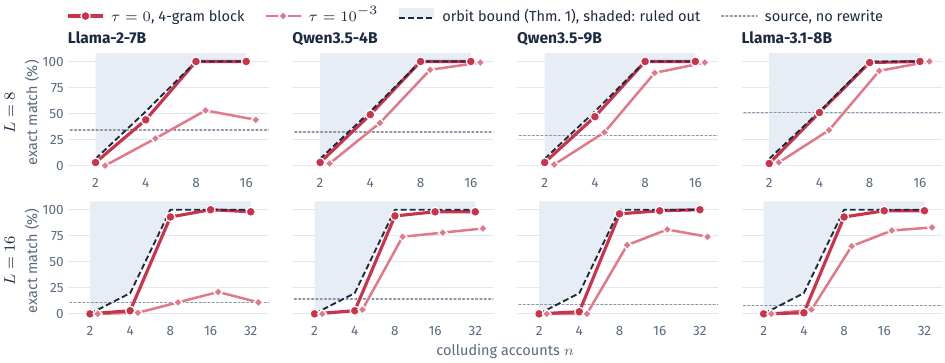}
\caption{StealthInk with its released 2-bit symbols: per-generator version
of Figure~\ref{fig:collusion}.  A target becomes reachable once every one of
its symbols has been seen at its position by some account.}
\label{fig:collusion-stealthink}
\end{figure*}

\begin{figure*}[t]
\centering
\includegraphics[width=\textwidth]{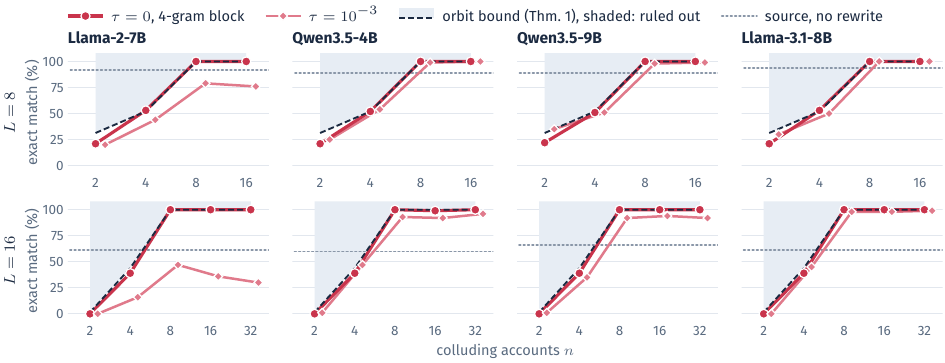}
\caption{XMark with its released 2-bit symbols: per-generator version of
Figure~\ref{fig:collusion}.  Blocks are public, so a target's symbol at a
block is reachable once some account carries it or only one symbol value is
unseen there.}
\label{fig:collusion-xmark}
\end{figure*}

\begin{figure*}[t]
\centering
\includegraphics[width=\textwidth]{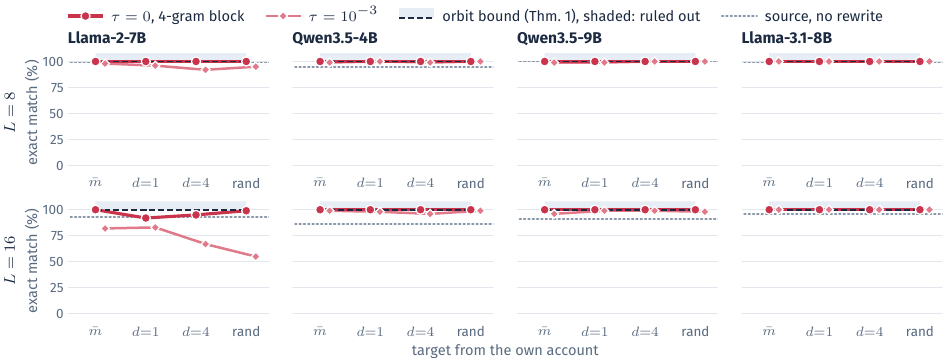}
\caption{CycleShift: per-generator version of Figure~\ref{fig:collusion}.
Every target is reachable from the own account, so the panels show the
complement $\bar m$, targets one and four bits away from $m$, and a uniformly
random one.}
\label{fig:collusion-cycleshift}
\end{figure*}

\section{Text length}
\label{app:tlen}

Figure~\ref{fig:tlen} repeats the sweep of Figure~\ref{fig:collusion} for
sources of 100, 200, and 400 tokens under the $\tau=0$ rule.  The symmetry bound
does not depend on the text length.  What changes is how many tokens each
message position receives.  Up to 16 bits the three lengths give the same
curves within a few points, except StealthInk at 16 bits and 100 tokens,
where eight to 32 accounts reach 74 to 79 percent instead of 94 to 100.  At
32 bits doubling the text to 400 tokens lifts the plateau of MPAC from 65 to
87 percent, StealthInk from 55 to 94, and XMark and RSBH to 100.  At 64 bits no scheme
carries its message through 200 tokens, untouched or rewritten (at most 23
percent, for XMark), and 400 tokens are still too few for MPAC (about 20
percent from 12 accounts on) and StealthInk (28 percent), while XMark reaches
82 to 85 percent, and the untouched watermarks decode none of these sources
exactly.

\begin{figure*}[t]
\centering
\includegraphics[width=\textwidth]{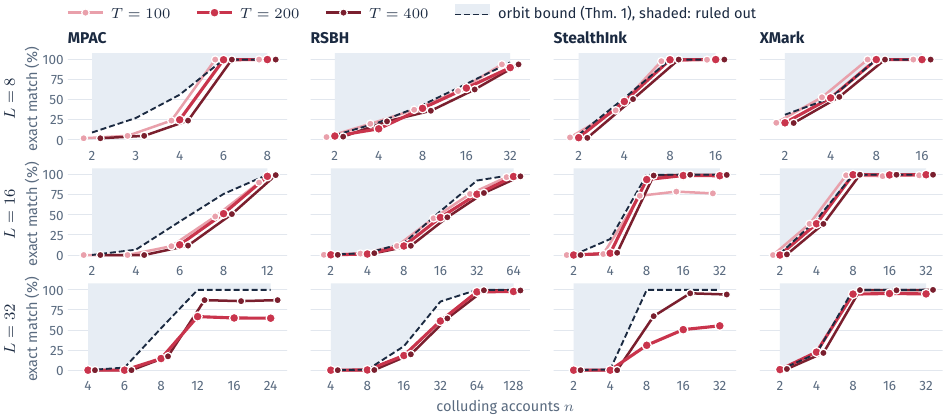}
\caption{Exact-match success of the constrained rewrite ($\tau=0$, 4-gram
block) against the number of colluding accounts for sources of $T\in\{100,
200, 400\}$ tokens, rewrites of $0.75T$ to $1.24T$ tokens, mean over the
four generators.  Dashed: the mean symmetry bound, the same for every $T$.
Shaded: ruled out.
Figures~\ref{fig:schemes-t100} and~\ref{fig:schemes-t400} show the same
runs in the layout of Figure~\ref{fig:collusion}.}
\label{fig:tlen}
\end{figure*}

\begin{figure*}[t]
\centering
\includegraphics[width=\textwidth]{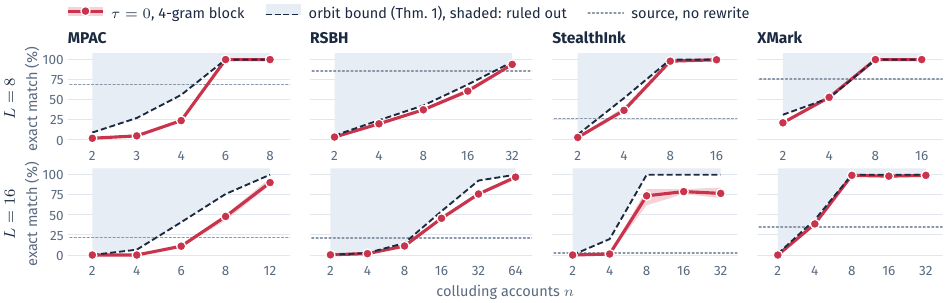}
\caption{Figure~\ref{fig:collusion} at 100 tokens per text ($\tau=0$ rule, since
the $\tau=10^{-3}$ rule and CycleShift were run at 200 tokens only).  Dotted:
the decoder's exact-match rate on the untouched 100-token source.}
\label{fig:schemes-t100}
\end{figure*}

\begin{figure*}[t]
\centering
\includegraphics[width=\textwidth]{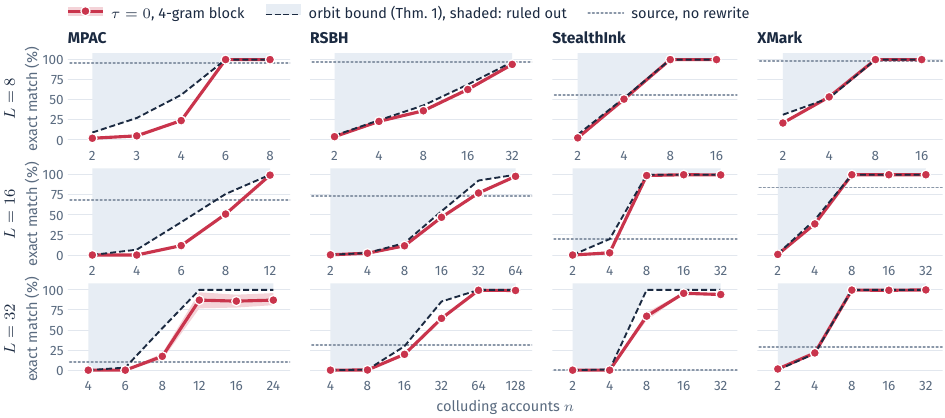}
\caption{Figure~\ref{fig:collusion} at 400 tokens per text ($\tau=0$ rule, since
the $\tau=10^{-3}$ rule and CycleShift were run at 200 tokens only).  Dotted: the decoder's exact-match rate on the untouched
400-token source.}
\label{fig:schemes-t400}
\end{figure*}

\section{Runtime}
\label{app:runtime}

Every attack job rewrites the 100 texts of one corpus under every condition
of its row in Figure~\ref{fig:collusion} on a single NVIDIA H100 (94\,GB), with the generator, the paraphraser, and the attacker's copy of the
base model loaded together and batches of 8 texts.  The API is simulated
locally with the generator itself, so the time includes the oracle's own
forward passes.  Table~\ref{tab:runtime} reports, for the settings run on the H100, the
wall-clock time per rewritten text under the $\tau=0$ rule, including
decoding, the certificate, and the learner's queries.  The $\tau=10^{-3}$ rule costs the same.  The
dominant cost is one base-model forward pass per new context, so schemes
whose context is longer than one token (XMark, StealthInk) query at almost
every step, while MPAC, CycleShift, and RSBH reuse contexts.  Queries is the
number of oracle calls per rewritten text with a plan, one per context and
account whose set the learner resolved, and corpus is the time to sample the 100
watermarked and 100 unwatermarked continuations of a corpus.

\begin{table*}[t]
\centering\footnotesize\setlength{\tabcolsep}{5pt}
\begin{tabular}{@{}lrrrrrrrr@{}}
\toprule
Scheme & $L$ & $T$ & \multicolumn{4}{c}{seconds per rewritten text $\downarrow$} & queries $\downarrow$ & corpus \\
\cmidrule(lr){4-7}
 & & & Llama-2-7B & Llama-3.1-8B & Qwen3.5-4B & Qwen3.5-9B & per rewrite & (min) \\
\midrule
MPAC & 8 & 200 & 1 & 2 & 3 & 3 & 120 & 3 \\
 & 16 & 200 & 1 & 2 & 3 & 3 & 174 & 3 \\
 & 32 & 200 & 2 & 2 & 4 & 5 & 845 & 2 \\
 & 64 & 200 & 2 & 3 & 5 & 5 & 1555 & 2 \\
 & 8 & 400 & 3 & 3 & 6 & 7 & 276 & 5 \\
 & 16 & 100 & 1 & 1 & 2 & 3 & 177 & 1 \\
 & 16 & 400 & 3 & 3 & 6 & 7 & 513 & 5 \\
 & 32 & 400 & 4 & 4 & 7 & 9 & 1585 & 5 \\
 & 64 & 400 & 5 & 5 & 8 & 9 & 2751 & 4 \\
\addlinespace
RSBH & 8 & 200 & 2 & 3 & 5 & 7 & 998 & 5 \\
 & 32 & 200 & 5 & 6 & 8 & 9 & 6520 & 4 \\
 & 8 & 400 & 3 & 4 & 8 & 9 & 1809 & 7 \\
 & 16 & 100 & 1 & 2 & 3 & 4 & 1444 & 2 \\
 & 16 & 400 & 5 & 5 & 10 & 14 & 4512 & 7 \\
 & 32 & 400 & 8 & 9 & 12 & 15 & 11364 & 7 \\
\addlinespace
StealthInk & 16 & 200 & 4 & 4 & 14 & 15 & 3038 & 3 \\
 & 32 & 200 & 4 & 4 & 14 & 14 & 3008 & 3 \\
 & 64 & 200 & 5 & 6 & 16 & 16 & 6450 & 3 \\
 & 8 & 100 & 2 & 2 & 6 & 7 & 642 & 3 \\
 & 8 & 400 & 8 & 8 & 32 & 46 & 3326 & 6 \\
 & 16 & 100 & 2 & 2 & 7 & 7 & 1443 & 2 \\
 & 16 & 400 & 9 & 9 & 33 & 38 & 6172 & 6 \\
 & 32 & 400 & 9 & 9 & 32 & 37 & 6125 & 6 \\
 & 64 & 400 & 11 & 11 & 37 & 39 & 12529 & 5 \\
\addlinespace
XMark & 16 & 200 & 2 & 2 & 9 & 9 & 146 & 3 \\
 & 32 & 200 & 2 & 2 & 9 & 9 & 156 & 3 \\
 & 64 & 200 & 2 & 2 & 9 & 9 & 154 & 3 \\
 & 8 & 400 & 4 & 4 & 15 & 18 & 302 & 6 \\
 & 16 & 100 & 1 & 1 & 5 & 5 & 80 & 2 \\
 & 16 & 400 & 4 & 4 & 16 & 17 & 263 & 6 \\
 & 32 & 400 & 4 & 4 & 16 & 18 & 280 & 6 \\
 & 64 & 400 & 4 & 4 & 17 & 18 & 275 & 6 \\
\addlinespace
CycleShift & 8 & 200 & 1 & 1 & 3 & 3 & 25 & 3 \\
 & 32 & 200 & 1 & 1 & 2 & 2 & 29 & 3 \\
\bottomrule
\end{tabular}

\caption{Wall-clock cost of the attack per rewritten text, in seconds on one
H100, for the $\tau=0$ rule, the number of oracle queries
per rewrite, and the corpus-generation time in minutes
(mean over generators).  $\downarrow$: lower is better.}
\label{tab:runtime}
\end{table*}

\section{Fidelity per generator}
\label{app:fidelity}

\begin{table*}[t]
\centering\footnotesize\setlength{\tabcolsep}{4pt}
\begin{tabular}{@{}llrrrrrrrr@{}}
\toprule
 & & \multicolumn{4}{c}{exact target match (\%) $\uparrow$} & & & \multicolumn{2}{c}{facts kept $\uparrow$} \\
\cmidrule(lr){3-6}\cmidrule(lr){9-10}
Generator & Text & own, $8$ & own, $16$ & $n$, $8$ & $n$, $16$ & judges $\uparrow$ & NLI $\uparrow$ & numbers & entities \\
\midrule
Llama-2-7B & paraphrase &  &  &  &  & 4.21 & 0.74 & 0.85 & 0.83 \\
 & $\tau=0$, block incl.\ source & 100 & 100 & 100 & 99 & 2.76 & 0.51 & 0.07 & 0.21 \\
 & $\tau=0$, rewrite-only block & 100 & 100 & 100 & 99 & 2.76 & 0.50 & 0.08 & 0.24 \\
 & \quad + fact protection & 100 & 90 & 100 & 93 & 3.10 & 0.59 & 0.66 & 0.56 \\
\addlinespace
Llama-3.1-8B & paraphrase &  &  &  &  & 4.17 & 0.75 & 0.80 & 0.77 \\
 & $\tau=0$, block incl.\ source & 100 & 100 & 100 & 99 & 3.10 & 0.60 & 0.13 & 0.34 \\
 & $\tau=0$, rewrite-only block & 100 & 100 & 100 & 100 & 3.08 & 0.60 & 0.15 & 0.37 \\
 & \quad + fact protection & 100 & 95 & 100 & 95 & 3.53 & 0.67 & 0.71 & 0.69 \\
\addlinespace
Qwen3.5-4B & paraphrase &  &  &  &  & 4.32 & 0.79 & 0.72 & 0.85 \\
 & $\tau=0$, block incl.\ source & 100 & 96 & 100 & 96 & 2.26 & 0.48 & 0.05 & 0.32 \\
 & $\tau=0$, rewrite-only block & 100 & 94 & 100 & 97 & 2.29 & 0.48 & 0.06 & 0.34 \\
 & \quad + fact protection & 98 & 88 & 98 & 89 & 2.90 & 0.61 & 0.60 & 0.66 \\
\addlinespace
Qwen3.5-9B & paraphrase &  &  &  &  & 4.54 & 0.81 & 0.80 & 0.89 \\
 & $\tau=0$, block incl.\ source & 100 & 97 & 100 & 98 & 2.51 & 0.55 & 0.08 & 0.35 \\
 & $\tau=0$, rewrite-only block & 100 & 98 & 100 & 100 & 2.54 & 0.54 & 0.06 & 0.38 \\
 & \quad + fact protection & 99 & 86 & 100 & 94 & 3.12 & 0.67 & 0.66 & 0.71 \\
\bottomrule
\end{tabular}

\caption{Per-generator fidelity results for the plain paraphrase, the rule of
Table~\ref{tab:models} ($\tau=0$, 4-gram block over
the whole input), the same rule with the block over the rewrite only, and the
latter with fact protection ($\kappa=0.5$).  Judges: mean of the two LLM judges.  NLI: mean of
support and coverage.  $\uparrow$: higher is better.}
\label{tab:fidelity-gen}
\end{table*}

\begin{figure*}[t]
\centering
\includegraphics[width=\textwidth]{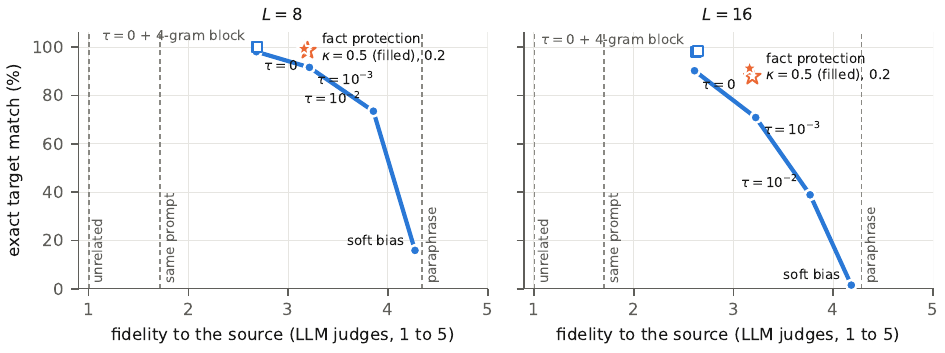}
\caption{Exact-match success against fidelity (mean of the two LLM judges)
for the rewrite rules evaluated in this section, one panel per message
length, mean over the four generators and over the own-account and $n$-account
rewrites.  Blue: the constrained rule from the soft
bias to $\tau=0$.  Squares: $\tau=0$ with either 4-gram block.  Stars: $\tau=0$
with fact protection.  Dashed lines: fidelity of the plain paraphrase, of an
unwatermarked completion of the same prompt, and of an unrelated source.}
\label{fig:fidelity}
\end{figure*}

\clearpage
\makeatletter\global\vsize\@colht\makeatother

\section{Relabelings and the general ambiguity bound}
\label{app:orbit-proof}

The main text avoids group notation and states the binary count directly in
Theorem~\ref{thm:targeting-bound}.  This appendix gives the general symmetry
argument behind that result.  It also makes explicit the assumptions on the
hidden-key distribution and on the decoder's treatment of labels, including
ties.

Let $\mathcal K$ be the key space and $\mathcal M$ the message space.  A
\emph{relabeling group} $(\mathcal G,\circ)$ acts on both spaces: every
$g\in\mathcal G$ induces maps
$\rho_g^{\mathcal K}\colon\mathcal K\to\mathcal K$ and
$\rho_g^{\mathcal M}\colon\mathcal M\to\mathcal M$.  We write both as
$\rho_g(x)$ because the argument identifies the space.  With identity $e$
and inverse $g^{-1}$, the action satisfies
\begin{equation}
 \begin{aligned}
  \rho_e(x)&=x,\\
  \rho_{g\circ h}(x)&=\rho_g\bigl(\rho_h(x)\bigr),\\
  G_{\rho_g(K)}^{\rho_g(m)}(c)&=G_K^m(c).
 \end{aligned}
 \label{eq:equivariance}
\end{equation}
The first two identities hold for $x\in\mathcal K\cup\mathcal M$ and the
last for every key, message, relabeling, and context.  Thus renaming the
hidden key structure and the message together leaves every green list
unchanged.

For the controlled messages $M=(m^{(1)},\ldots,m^{(n)})$, define the
relabelings that preserve all their public names by
\begin{equation}
 \mathcal H_M:=\{g\in\mathcal G:\rho_g(m^{(i)})=m^{(i)}
                         \text{ for every }i\}.
 \label{eq:stabilizer}
\end{equation}
The messages indistinguishable from a target $m'$ under these remaining
symmetries form
\begin{equation}
 \operatorname{Orb}_{\mathcal H_M}(m')
 :=\{\rho_g(m'):g\in\mathcal H_M\}.
 \label{eq:orbit}
\end{equation}
Appendix~\ref{app:relabeling-examples} gives concrete examples of a
relabeling that remains in $\mathcal H_M$ and one that the controlled
messages exclude.

\paragraph{Full symmetry statement for Theorem~\ref{thm:targeting-bound}.}
Fix the controlled messages $M=(m^{(1)},\ldots,m^{(n)})$, a target $m'$,
and a finite relabeling group satisfying \eqref{eq:equivariance}. Write
$H:=\mathcal H_M$ and assume the following.
\begin{enumerate}[label=\textnormal{(A\arabic*)},leftmargin=*,labelsep=0.3em]
\item \textbf{Invariant key prior.}
The key $K$ has a probability distribution $\mu$ satisfying
\begin{equation}
 \mu(\rho_g(B))=\mu(B)
 \qquad(g\in H)
 \label{eq:orbit-prior}
\end{equation}
for every measurable set of keys $B\subseteq\mathcal K$, where
$\rho_g(B):=\{\rho_g(k):k\in B\}$. Thus a set of keys and its relabeling
have the same prior probability.
\item \textbf{Observations through the controlled accounts.}
The transcript $O$ records all queries and responses available to the attack.
The accounts generate responses according to (P1), and the attack may choose
accounts, contexts, prompts, and when to stop from its earlier observations.
Exact next-token distributions are also allowed as responses.
Sample responses use fresh sampling randomness independent of the key
and of the attack's private randomness, with the conditional token laws
specified by (P1).
All queries finish and the attack outputs a finite text almost surely.
There is no additional key-dependent information outside $O$.
The attack's private random seed $R$, containing all of its random choices,
is independent of $K$ before observations are collected. Its output is
$A(O,R)$, and the notation $A(O)$ in the theorem leaves $R$ implicit.
\item \textbf{Equivariant decoder.}
The decoder returns a message or an abstention symbol $\bot$, and for every
key $k$, text $z$, message $b$, and $g\in H$,
\begin{equation}
 \begin{aligned}
 &\Pr[D_{\rho_g(k)}(z)=\rho_g(b)]\\
 &\hspace{25pt}=\Pr[D_k(z)=b].
 \end{aligned}
 \label{eq:orbit-decoder-kernel}
\end{equation}
These probabilities are only over the decoder's own fresh randomness,
independent of the attack and of generation. For a deterministic decoder,
this is the identity $D_{\rho_g(k)}(z)=\rho_g(D_k(z))$.
We extend the action by $\rho_g(\bot)=\bot$.
\end{enumerate}
Public model and sampling settings are fixed throughout, and independent public
randomness can be included in $R$. Let
\begin{equation}
 \mathcal B:=\operatorname{Orb}_{H}(m'),\qquad N:=|\mathcal B|<\infty.
 \label{eq:orbit-proof-size}
\end{equation}
Then
\begin{equation}
 \Pr[D_K(A(O,R))=m'\mid O]\le \frac1N.
 \label{eq:orbit-proof-bound}
\end{equation}
The inequality holds almost surely in $O$. The probability includes the
key, account sampling, attack randomness, and
any decoder randomness. ``Almost surely'' means that the inequality may
exclude observations in a set of total probability zero. In particular, it
holds at every discrete observation $o$ with $\Pr[O=o]>0$.

\begin{proof}
We first show that even the attack's complete view, including its seed,
leaves the key distribution invariant under $H$. We then show that the
decoder assigns equal conditional probabilities to every message in
$\mathcal B$, and sum those probabilities.

\smallskip\noindent
\textbf{Step 1: check the group operations used below.}
The action identities in \eqref{eq:equivariance} give, for any key or
message $x$,
\begin{equation}
 \begin{aligned}
 \rho_{g^{-1}}(\rho_g(x))
 &=\rho_{g^{-1}\circ g}(x)=x,\\
 \rho_g(\rho_{g^{-1}}(x))
 &=\rho_{g\circ g^{-1}}(x)=x.
 \end{aligned}
\end{equation}
Thus $\rho_g$ is a bijection with inverse $\rho_{g^{-1}}$.
The identity fixes every controlled message, so $e\in H$.
If $g,h\in H$, then for each $i$,
\begin{equation}
 \rho_{g\circ h}(m^{(i)})
 =\rho_g(\rho_h(m^{(i)}))
 =\rho_g(m^{(i)})=m^{(i)}.
\end{equation}
Hence $g\circ h\in H$. If $g\in H$, applying $\rho_{g^{-1}}$ to
$\rho_g(m^{(i)})=m^{(i)}$ gives
$\rho_{g^{-1}}(m^{(i)})=m^{(i)}$, so $g^{-1}\in H$ as well.
This verifies that $H$ is a subgroup. It also gives $m'=\rho_e(m')\in
\mathcal B$, so $N\ge1$ and division by $N$ is valid.

\smallskip\noindent
\textbf{Step 2: show that each account response is unchanged.}
Fix $g\in H$, a key $k$, an account $i$, and a context $c$.
The stabilizer condition and \eqref{eq:equivariance} give
\begin{equation}
 \begin{aligned}
 G_{\rho_g(k)}^{m^{(i)}}(c)
 &=G_{\rho_g(k)}^{\rho_g(m^{(i)})}(c)\\
 &=G_k^{m^{(i)}}(c).
 \end{aligned}
 \label{eq:orbit-fixed-account-set}
\end{equation}
The first equality uses $\rho_g(m^{(i)})=m^{(i)}$, and the second relabels
the key and message together. By (P1), for every token $v$,
\begin{equation}
 \begin{aligned}
 p_{\rho_g(k),m^{(i)}}(v\mid c)
 &=\Phi\bigl(p(\cdot\mid c),G_{\rho_g(k)}^{m^{(i)}}(c)\bigr)(v)\\
 &=\Phi\bigl(p(\cdot\mid c),G_k^{m^{(i)}}(c)\bigr)(v)\\
 &=p_{k,m^{(i)}}(v\mid c).
 \end{aligned}
 \label{eq:orbit-fixed-account-law}
\end{equation}
Here $\Phi(p,S)(v)$ denotes the probability of $v$ in the distribution
returned by $\Phi$. The base distribution and the rule $\Phi$ are the
same in both experiments because they are key-independent.

Consequently, a query returning the entire next-token distribution has
exactly the same answer. A query sampling a token has the same answer
distribution. For a completion $w=(w_1,\ldots,w_T)$, fix the query prompt
and let $c_t(w_{<t})$ be the context determined by that prompt and prefix.
The probability of this completion is
\begin{equation}
 \begin{aligned}
 &p_{\rho_g(k),m^{(i)}}(w)\\
 &=\prod_{t=1}^T
 p_{\rho_g(k),m^{(i)}}(w_t\mid c_t(w_{<t}))\\
 &=\prod_{t=1}^T
 p_{k,m^{(i)}}(w_t\mid c_t(w_{<t}))\\
 &=p_{k,m^{(i)}}(w).
 \end{aligned}
\end{equation}
The first and last equalities are the autoregressive sampling rule, and the
middle equality applies \eqref{eq:orbit-fixed-account-law} to each factor.
The identity includes zero-probability sequences. If the response stops
at an end token or a public length limit, the same stopping rule is applied
to the same sequence, so the resulting response laws still agree.

\smallskip\noindent
\textbf{Step 3: preserve this equality under adaptive queries.}
Fix the private seed at $R=r$. Once this seed is fixed, each next query
and the decision to stop are deterministic functions of the preceding
transcript. Write $T_j$ for the transcript after $j$ queries, with
$T_0$ empty. For a possible preceding transcript $t$, denote the next
query by $q_j(r,t)$. Denote its response distribution under key $k$ by
$Q_k(\cdot\mid q_j(r,t),t)$. Step 2 gives, for every response event $F$,
\begin{equation}
 \begin{aligned}
 &Q_{\rho_g(k)}(F\mid q_j(r,t),t)\\
 &\hspace{20pt}=Q_k(F\mid q_j(r,t),t).
 \end{aligned}
 \label{eq:orbit-response-kernel}
\end{equation}
For ordinary fresh sampling this kernel is the distribution from Step 2.
Writing the history explicitly also accommodates repeated contexts and
queries chosen using earlier responses. We do not require different
queries to concern different contexts.

We prove equality of transcript laws by induction on $j$.
It holds at $j=0$ because both transcripts are empty.
Suppose it holds after $j-1$ queries, and call the common law $\eta_{j-1}$.
For an event $E$ of extended transcripts, the probability after the next
query under key $k$ is
\begin{equation}
 \begin{aligned}
 &\int\!\int
 \mathbf 1\{(t,q_j(r,t),a)\in E\}\\
 &\qquad\cdot Q_k(\mathrm da\mid q_j(r,t),t)
 \,\eta_{j-1}(\mathrm dt).
 \end{aligned}
 \label{eq:orbit-transcript-induction}
\end{equation}
The expression under key $\rho_g(k)$ has the same outer measure by the
induction hypothesis, the same query because $r,t$ are fixed, and the
same inner measure by \eqref{eq:orbit-response-kernel}. The two
probabilities are therefore equal, completing the induction.

Stopping causes no exception: at a fixed $r,t$, both experiments make
the same stopping decision. Equivalently, append a fixed stop symbol
at every subsequent round once the attack stops. The preceding induction
then applies to these padded transcripts. Since termination occurs after
finitely many queries almost surely, every event of final transcripts is
the disjoint union of its parts with stopping times $0,1,2,\ldots$.
Equality for each part gives equality for their sum. Thus, for every
transcript event $E$,
\begin{equation}
 \begin{aligned}
 &\Pr[O\in E\mid K=\rho_g(k),R=r]\\
 &\hspace{20pt}=\Pr[O\in E\mid K=k,R=r].
 \end{aligned}
 \label{eq:orbit-transcript-law}
\end{equation}
This proves likelihood invariance for the whole adaptive experiment,
including exact distribution responses, whose laws are point masses.

\smallskip\noindent
\textbf{Step 4: derive invariance of the conditional key distribution.}
Let $V:=(O,R)$ be the attack's complete view, including the private seed.
Although $R$ and $K$ are independent initially, they need not remain
independent after conditioning on $O$. We therefore retain $R$ in $V$.

For an event $C$ of complete views, put
$L_C(k):=\Pr[V\in C\mid K=k]$. Let $\nu$ be the initial law of $R$ and
write $C_r:=\{o:(o,r)\in C\}$. Independence before querying gives
\begin{equation}
 L_C(k)=\int
 \Pr[O\in C_r\mid K=k,R=r],\nu(\mathrm dr).
\end{equation}
Equation~\eqref{eq:orbit-transcript-law} holds for each fixed seed, so
integrating it against this same $\nu$ yields
\begin{equation}
 L_C(\rho_g(k))=L_C(k).
 \label{eq:orbit-view-likelihood}
\end{equation}
For a set of keys $B$, the definition of conditional probability gives
\begin{equation}
 \begin{aligned}
 \Pr[K\in\rho_g(B),V\in C]
 &=\int_{\rho_g(B)}L_C(k)\,\mu(\mathrm dk)\\
 &=\int_B L_C(\rho_g(k))\,\mu(\mathrm dk)\\
 &=\int_B L_C(k)\,\mu(\mathrm dk)\\
 &=\Pr[K\in B,V\in C].
 \end{aligned}
 \label{eq:orbit-joint-invariance}
\end{equation}
The second line changes variables through the bijection $\rho_g$,
and (A1) says that this change preserves $\mu$. The third line uses
\eqref{eq:orbit-view-likelihood}.

For a discrete view $v$ with positive probability, set $C=\{v\}$ in
\eqref{eq:orbit-joint-invariance} and divide by $\Pr[V=v]$. This gives
\begin{equation}
 \Pr[K\in\rho_g(B)\mid V=v]
 =\Pr[K\in B\mid V=v].
 \label{eq:orbit-posterior-invariance}
\end{equation}
For general observations, including real-valued distribution responses,
we use a regular conditional distribution
$\mu_v(B):=\Pr[K\in B\mid V=v]$.
To see why the same conclusion holds, let $P_V$ be the law of $V$ and
rewrite \eqref{eq:orbit-joint-invariance} as
\begin{equation}
 \begin{aligned}
 &\int_C\mu_v(\rho_g(B))\,P_V(\mathrm dv)\\
 &\qquad=\int_C\mu_v(B)\,P_V(\mathrm dv).
 \end{aligned}
\end{equation}
This holds for every view event $C$. Two measurable functions with equal
integrals over every event $C$ are
equal almost everywhere: otherwise, on an event where one exceeds the
other by at least $1/j$ for some integer $j\ge1$, their integrals would
differ. Applying this fact gives \eqref{eq:orbit-posterior-invariance}
almost surely in $v$. The key and observation spaces here are finite or
countable products of discrete and real-valued spaces, on which these
conditional distributions exist. Equality can first be imposed on a
countable generating algebra of key events, then extended to all measurable
events by uniqueness of probability measures. Since $H$ is finite, one
common set of views of probability one works for every $g\in H$.

In particular, at each such view and for every bounded measurable function
$f$ of the key,
\begin{equation}
 \int f(\rho_g(k))\,\mu_v(\mathrm dk)
 =\int f(k)\,\mu_v(\mathrm dk).
 \label{eq:orbit-posterior-integral}
\end{equation}
Indeed, \eqref{eq:orbit-posterior-invariance}, applied to $g^{-1}\in H$,
says that $\rho_g(K)$ and $K$ have the same conditional distribution.
Expectations of the same function under equal distributions coincide.

\smallskip\noindent
\textbf{Step 5: make the decoder's symmetry explicit, including ties.}
For a fixed key and text, define the decoder probabilities
\begin{equation}
 d_k(z,b):=\Pr[D_k(z)=b].
 \label{eq:orbit-decoder-probability}
\end{equation}
If the decoder is deterministic, $d_k(z,b)$ is the indicator of
$D_k(z)=b$. Assumption (A3) states that
\begin{equation}
 d_{\rho_g(k)}(z,\rho_g(b))=d_k(z,b).
 \label{eq:orbit-decoder-equivariant}
\end{equation}
We explain how this condition follows for a common membership-based
scoring rule, and why its tie rule matters.

For a fixed text $z=(z_1,\ldots,z_T)$, let $c_t$ be its scoring contexts
and put
\begin{equation}
 \beta_k(z,b):=
 \bigl(\mathbf 1\{z_t\in G_k^b(c_t)\}\bigr)_{t=1}^T.
\end{equation}
For each $t$, \eqref{eq:equivariance} gives
\begin{equation}
 \mathbf 1\{z_t\in G_{\rho_g(k)}^{\rho_g(b)}(c_t)\}
 =\mathbf 1\{z_t\in G_k^b(c_t)\}.
\end{equation}
Equality at every coordinate gives
\begin{equation}
 \beta_{\rho_g(k)}(z,\rho_g(b))=\beta_k(z,b).
\end{equation}
Suppose each candidate is evaluated by the same public function,
\begin{equation}
 s_k(z,b):=s(z,\beta_k(z,b)).
\end{equation}
Substitution then yields
\begin{equation}
 s_{\rho_g(k)}(z,\rho_g(b))=s_k(z,b).
 \label{eq:orbit-score-equivariance}
\end{equation}
The use of the same function for every candidate is part of this sufficient
condition.  A candidate-dependent public preference is not justified by
key access through memberships alone.

For a maximizing decoder, write
\begin{equation}
 \mathcal E_k(z):=\arg\max_{b\in\mathcal M}s_k(z,b).
\end{equation}
If $b\in\mathcal E_k(z)$, then $s_k(z,b)\ge s_k(z,a)$ for every
$a\in\mathcal M$. Equation~\eqref{eq:orbit-score-equivariance} implies
\begin{equation}
 s_{\rho_g(k)}(z,\rho_g(b))
 \ge s_{\rho_g(k)}(z,\rho_g(a))
 \quad\text{for every }a.
\end{equation}
Because $\rho_g$ is a bijection on messages, $\rho_g(a)$ ranges over all
candidates. Thus $\rho_g(b)\in\mathcal E_{\rho_g(k)}(z)$.
This proves one inclusion in
\begin{equation}
 \mathcal E_{\rho_g(k)}(z)=\rho_g(\mathcal E_k(z)).
 \label{eq:orbit-extrema}
\end{equation}
For the reverse inclusion, take $b'\in\mathcal E_{\rho_g(k)}(z)$ and
apply the same argument to the key $\rho_g(k)$ and the element $g^{-1}$.
It gives $\rho_{g^{-1}}(b')\in\mathcal E_k(z)$, hence
$b'\in\rho_g(\mathcal E_k(z))$ as required.
For a minimizing decoder, replace each score inequality above by
$\le$, and both inclusions follow in exactly the same way.

If the extremum is unique, \eqref{eq:orbit-extrema} relabels that unique
output, proving \eqref{eq:orbit-decoder-equivariant}.
For a finite nonempty tied set, uniform tie breaking gives
\begin{equation}
 d_k(z,b)=
 \frac{\mathbf 1\{b\in\mathcal E_k(z)\}}{|\mathcal E_k(z)|}.
\end{equation}
Equation~\eqref{eq:orbit-extrema} preserves membership in the tied set
and its cardinality, so it again gives
\eqref{eq:orbit-decoder-equivariant}. If instead the decoder abstains
whenever there is a tie, both relabeled experiments abstain, since their
tied sets have the same size. With $\rho_g(\bot)=\bot$, this also
satisfies (A3). A deterministic preference for the largest or smallest
label need not do so. For other decoder architectures, (A3) is the
condition required, since (P2) alone does not establish their tie behavior.

\smallskip\noindent
\textbf{Step 6: show equality of conditional decoding probabilities.}
Fix a complete view $v=(o,r)$ in the probability-one set from Step 4.
The attack's rewritten text
\begin{equation}
 z:=A(o,r)
\end{equation}
is now fixed. Relabeling a key in its conditional distribution does
not change $o$, $r$, or this text. For every message $b$, put
\begin{equation}
 \begin{aligned}
 p_v(b)
 &:=\Pr[D_K(A(O,R))=b\mid V=v]\\
 &=\int d_k(z,b)\,\mu_v(\mathrm dk).
 \end{aligned}
 \label{eq:orbit-output-posterior}
\end{equation}
The second line averages the decoder's probabilities over the conditional
key distribution. It uses the fact that the decoder's fresh randomness
is independent of the view once its key and input are fixed.

For any $g\in H$, compute
\begin{equation}
 \begin{aligned}
 p_v(\rho_g(b))
 &=\int d_k(z,\rho_g(b))\,\mu_v(\mathrm dk)\\
 &=\int d_{\rho_g(k)}(z,\rho_g(b))\,\mu_v(\mathrm dk)\\
 &=\int d_k(z,b)\,\mu_v(\mathrm dk)\\
 &=p_v(b).
 \end{aligned}
 \label{eq:orbit-equal-probabilities}
\end{equation}
The first and last lines use \eqref{eq:orbit-output-posterior}.
The second line applies \eqref{eq:orbit-posterior-integral} to the
bounded function $f(k)=d_k(z,\rho_g(b))$.
The third line applies \eqref{eq:orbit-decoder-equivariant} at each key.
Notice that the target is not relabeled in the actual attack.  The
calculation compares the probabilities of different decoder outputs
for the same fixed text $z$.

\smallskip\noindent
\textbf{Step 7: sum over the distinct messages in the orbit.}
By definition of $\mathcal B$, each $b\in\mathcal B$ has the form
$b=\rho_g(m')$ for some $g\in H$. Substituting $m'$ for $b$ in
\eqref{eq:orbit-equal-probabilities} therefore gives
$p_v(b)=p_v(m')$ for every $b\in\mathcal B$.
There are $N$ distinct messages in this set. The events that the
decoder outputs each of these messages are pairwise disjoint: one
decoder run returns at most one message. Consequently,
\begin{equation}
 \begin{aligned}
 Np_v(m')
 &=\sum_{b\in\mathcal B}p_v(b)\\
 &=\Pr[D_K(z)\in\mathcal B\mid V=v]\\
 &\le1.
 \end{aligned}
 \label{eq:orbit-sum}
\end{equation}
We sum over distinct messages, not group elements: several relabelings
can send $m'$ to the same message and must not be counted repeatedly.
Dividing by $N\ge1$ proves
\begin{equation}
 \Pr[D_K(A(O,R))=m'\mid O,R]\le\frac1N.
 \label{eq:orbit-full-view-bound}
\end{equation}
This inequality holds almost surely in $(O,R)$.
The inequality in \eqref{eq:orbit-sum} allows the decoder to output
messages outside $\mathcal B$ or to abstain. Thus we have not assumed
that its total probability on this orbit is one.

\smallskip\noindent
\textbf{Step 8: average over the attack's private randomness.}
Let $E:=\{D_K(A(O,R))=m'\}$ be the success event.
Because conditioning on $(O,R)$ includes conditioning on $O$, the law
of iterated conditional expectation gives
\begin{equation}
 \begin{aligned}
 \Pr[E\mid O]
 &=\mathbb E[\mathbf 1_E\mid O]\\
 &=\mathbb E\!\left[
      \mathbb E[\mathbf 1_E\mid O,R]\,\middle|\,O\right]\\
 &=\mathbb E\!\left[\Pr[E\mid O,R]\,\middle|\,O\right]\\
 &\le\mathbb E\!\left[\frac1N\,\middle|\,O\right]\\
 &=\frac1N.
 \end{aligned}
\end{equation}
The inequality substitutes \eqref{eq:orbit-full-view-bound} and uses
monotonicity of conditional expectation. The last equality holds
because $M$ and $m'$ are fixed, so their orbit size $N$ is fixed.
This averaging uses the actual conditional law of $R$ given $O$,
and it does not assume that the seed remains independent of the key
after observation. Replacing $N$ by its definition in
\eqref{eq:orbit-proof-size} gives the claimed general bound.
\end{proof}

\paragraph{When the ideal-key assumption applies.}
Assumption (A1) is a distributional requirement, not a consequence of
the algebraic identity \eqref{eq:equivariance} alone. In the binary
position--color idealization, at each context the position is uniform
on $[L]$ and the balanced two-color partition is uniform, independently
of that position and of the assignments at other contexts. A position
permutation and per-position color flips define a bijection of the
possible (position, partition) pairs. Each pair originally has the same
probability, so the relabeled pair is again uniform. The transformation
acts separately at each context, preserving independence across
contexts, and repeated contexts still read the same assignment.
This proves invariance of the entire ideal key table. For a concrete
finite-seed pseudorandom implementation, a uniform seed by itself does
not prove that the induced table distribution has this exact symmetry.
The information-theoretic conclusion above uses (A1).

\paragraph{Why the condition on ties cannot be dropped.}
Consider two-bit messages, a single controlled message $(0,1)$, and
the relabeling
\begin{equation}
 \rho_g(b_1,b_2)=(1-b_2,1-b_1).
\end{equation}
This is a position swap together with a flip at each position in the
binary group described in the main text. It fixes $(0,1)$, so $g\in H$,
but it interchanges $(0,0)$ and $(1,1)$. With one controlled message,
a stabilizing relabeling has its color flips determined by its position
permutation. There are just two position permutations here: the identity
gives the identity relabeling, and the swap gives $g$. Therefore
\begin{equation}
 \operatorname{Orb}_{H}((1,1))=\{(1,1),(0,0)\}.
\end{equation}
On an empty text all vote counts are zero. A decoder that resolves every
tied bit toward $1$ returns $(1,1)$ for every key. Its success probability
for that target is $1$, exceeding the purported bound $1/2$.
The equivariance assumption excludes precisely this failure.
The example concerns exact decoding, and it does not claim that an empty
text passes watermark acceptance or preserves a source's content.

\paragraph{Source texts and unlabelled observations.}
The attacker of Section~\ref{sec:problem-formulation} also observes the
source $y$.
To apply the theorem, include it in the observation:
\begin{equation}
 O=(y,O_M),\qquad A(O,R)=A(y,m',O_M,R).
\end{equation}
Here the fixed target is left implicit on the left-hand side.
If $y$ is generated under a controlled message, Step 2 applies directly.
If $y$ is generated independently of the key and the account experiment,
it also preserves the likelihood invariance.
A source under a different unknown message requires an
additional assumption, because simply not revealing its message does not imply
symmetry.

For example, let $U=(U_1,\ldots,U_s)$ be the unknown messages of all
additional watermarked texts. Define the joint relabeling
\begin{equation}
 (K,U)\longmapsto
 (\rho_g(K),\rho_g(U_1),\ldots,\rho_g(U_s)).
\end{equation}
Suppose their joint prior is invariant under this map for every $g\in H$.
An invariant key prior with independent uniform $U_j$ on a finite message
space is one sufficient choice: the map permutes that message space,
so it preserves each uniform distribution and their product.
For each additional text, \eqref{eq:equivariance} gives
\begin{equation}
 G_{\rho_g(k)}^{\rho_g(u_j)}(c)=G_k^{u_j}(c).
\end{equation}
Thus the response likelihood is invariant when key and hidden messages
are relabeled together. Step 3 applies to this enlarged hidden state,
and Step 4 gives invariance of its posterior. Taking a posterior event
of the form $\{K\in B,\ U\in\mathcal M^s\}$ then gives invariance of
the marginal posterior of $K$, since the full message space maps to
itself. Steps 5--8 therefore prove the same bound conditional on all
these observations. If an extra source message is known, it should
instead be added to $M$, and the orbit recomputed under the resulting
smaller stabilizer. An arbitrary fixed or nonuniform unknown source
message is not covered without the stated joint symmetry.

\paragraph{Content constraints, acceptance, and the meaning of the bound.}
The proof applies to every output text $A(O,R)$, so restricting the
attacker to content-preserving rewrites of $y$, as condition (ii) of its
goal in Section~\ref{sec:problem-formulation} demands, cannot invalidate
the upper bound when $y$ is included in $O$ as above. A deployment may also apply a keyed
acceptance test $\operatorname{Acc}_K(z)\in\{0,1\}$, which decides whether a
text is watermarked at all, before it attributes the text to the decoded
message, and the attacker's goal in Section~\ref{sec:problem-formulation}
omits this test. The bound
covers it nonetheless: for $Z:=A(O,R)$,
\begin{equation}
 \begin{aligned}
 &\Pr[D_K(Z)=m',\ \operatorname{Acc}_K(Z)=1\mid O]\\
 &\qquad\le\Pr[D_K(Z)=m'\mid O]\le\frac1N.
 \end{aligned}
\end{equation}
The first inequality follows because the event on the first line is
a subset of the exact-decoding event. It requires no symmetry of the
acceptance rule when $D_K$ itself satisfies (A3). This is a bound on
joint success and acceptance, not on the success probability conditional
on acceptance.

For a released decoder with a label-dependent tie rule, apply the theorem
instead to the decoder that abstains whenever the relevant extremum is
not unique. On texts without such ties the two decoders agree, so the
joint event of exact success and no ties is bounded by $1/N$.
To transfer that bound to accepted outputs of the released decoder,
one must establish that acceptance excludes all the relevant ties,
and the presence of an acceptance threshold alone does not establish this.

Finally, when $N=1$ the bound is $1$ and places no restriction on success.
When $N>1$, Step 7 identifies the obstruction: all $N$ orbit messages
have equal conditional decoding probability. The theorem supplies an
upper bound, not an attack attaining it.  In particular, a singleton
orbit alone does not prove that a content-preserving successful rewrite
exists.

\section{Binary ambiguity count}
\label{app:closed-orbit-proof}

We make the binary position--color setting of Section~\ref{sec:method}
explicit, verify that its relabelings form a group acting as
\eqref{eq:equivariance} requires, determine the stabilizer $\mathcal H_M$
of the controlled messages exactly, and count the orbit of the target.
Throughout, $(u+w)\bmod 2$ denotes the sum of the bits $u$ and $w$
reduced modulo~2, so that $(u+u)\bmod 2=0$.  Since
$\bigl((u+w)\bmod 2+z\bigr)\bmod 2=(u+w+z)\bmod 2$, a reduction inside a
longer sum may be postponed to the end, and we do so without comment.
On bit vectors the reduction acts coordinatewise.

\paragraph{Closed-form count.}
Fix the following.
\begin{enumerate}[label=\textnormal{(S\arabic*)},leftmargin=*,labelsep=0.3em]
\item \textbf{Messages, keys, and sets.}
The message space is $\mathcal M=\{0,1\}^L$.  A key $K$ assigns to every
context $c$ a position $J_K(c)\in[L]$ and a coloring
$C_K(c,\cdot)\colon\mathcal V\to\{0,1\}$, and the green list of message
$m$ at $c$ is $G_K^m(c)=\{v:C_K(c,v)=m_{J_K(c)}\}$, as in
\eqref{eq:binary-label-set}.
\item \textbf{Relabelings.}
An element of $\mathcal G$ is a pair $g=(\sigma,f)$ of a permutation
$\sigma\in\mathfrak S_L$ of the positions and a flip vector
$f\in\{0,1\}^L$.  It acts on messages and keys as follows, where the first
line holds for every $j\in[L]$,
\begin{equation}
 \begin{aligned}
  \bigl[\rho_g(m)\bigr]_{\sigma(j)}&=(m_j+f_j)\bmod 2,\\
  J_{\rho_g(K)}(c)&=\sigma\bigl(J_K(c)\bigr),\\
  C_{\rho_g(K)}(c,v)&=\bigl(C_K(c,v)+f_{J_K(c)}\bigr)\bmod 2.
 \end{aligned}
 \label{eq:app-co-action}
\end{equation}
Thus the bit at position $j$ moves to position $\sigma(j)$ and is flipped
exactly when $f_j=1$, and the key renames its positions by $\sigma$ and
exchanges the two color names at every context whose position $j$ has
$f_j=1$.  The identity is $e=(\mathrm{id},0)$, and the product of
$g=(\sigma_g,f^g)$ and $h=(\sigma_h,f^h)$ is
\begin{equation}
 \begin{aligned}
 g\circ h&=\bigl(\sigma_g\circ\sigma_h,\,f^{g\circ h}\bigr),\\
 f^{g\circ h}_j&=\bigl(f^h_j+f^g_{\sigma_h(j)}\bigr)\bmod 2.
 \end{aligned}
 \label{eq:app-co-product}
\end{equation}
\item \textbf{Signatures and twin classes.}
For the controlled messages $M=(m^{(1)},\ldots,m^{(n)})$ and every
position $j$, the account signature is $x_j=(x_{j,2},\ldots,x_{j,n})$ with
$x_{j,i}=(m^{(i)}_j+m^{(1)}_j)\bmod 2$, as in \eqref{eq:account-signature}.
Positions with equal signatures are twins, $\mathcal T$ is the partition of
$[L]$ into twin classes, and for the target $m'$ the relative bits are
$r_j=(m'_j+m^{(1)}_j)\bmod 2$, with $a_T=\sum_{j\in T}r_j$ for
$T\in\mathcal T$.
\end{enumerate}
Then $(\mathcal G,\circ)$ is a finite group, the action
\eqref{eq:app-co-action} satisfies \eqref{eq:equivariance}, and the orbit
of $m'$ under the stabilizer $\mathcal H_M$ of \eqref{eq:stabilizer} has
size
\begin{equation}
 |\operatorname{Orb}_{\mathcal H_M}(m')|=\prod_{T\in\mathcal T}\binom{|T|}{a_T}.
 \label{eq:app-co-claim}
\end{equation}

\begin{proof}
We first check that \eqref{eq:app-co-action} is a group action with the
equivariance property, so that the general bound of
Appendix~\ref{app:orbit-proof} applies.  We then determine
$\mathcal H_M$ exactly, pass to coordinates relative to the reference
message, and count.

\smallskip\noindent
\textbf{Step 1: verify the group action and the equivariance identity.}
The identity acts trivially: with $e=(\mathrm{id},0)$,
\eqref{eq:app-co-action} gives $[\rho_e(m)]_j=m_j$, $J_{\rho_e(K)}=J_K$,
and $C_{\rho_e(K)}=C_K$, so $\rho_e(x)=x$ for every message and key $x$.
For the composition, fix $g$, $h$, and a message $m$.  Applying
\eqref{eq:app-co-action} twice,
\begin{equation}
 \begin{aligned}
 &\bigl[\rho_g(\rho_h(m))\bigr]_{\sigma_g(\sigma_h(j))}\\
 &\quad=\bigl(\bigl[\rho_h(m)\bigr]_{\sigma_h(j)}+f^g_{\sigma_h(j)}\bigr)\bmod 2\\
 &\quad=\bigl(m_j+f^h_j+f^g_{\sigma_h(j)}\bigr)\bmod 2\\
 &\quad=\bigl(m_j+f^{g\circ h}_j\bigr)\bmod 2,
 \end{aligned}
\end{equation}
which is $[\rho_{g\circ h}(m)]_{(\sigma_g\circ\sigma_h)(j)}$ by
\eqref{eq:app-co-action} and \eqref{eq:app-co-product}.  For a key,
\begin{equation}
 \begin{aligned}
 &J_{\rho_g(\rho_h(K))}(c)
  =\sigma_g\bigl(J_{\rho_h(K)}(c)\bigr)\\
 &\quad=\sigma_g\bigl(\sigma_h(J_K(c))\bigr),\\
 &C_{\rho_g(\rho_h(K))}(c,v)\\
 &\quad=\bigl(C_{\rho_h(K)}(c,v)+f^g_{J_{\rho_h(K)}(c)}\bigr)\bmod 2\\
 &\quad=\bigl(C_K(c,v)+f^h_{J_K(c)}+f^g_{\sigma_h(J_K(c))}\bigr)\bmod 2\\
 &\quad=\bigl(C_K(c,v)+f^{g\circ h}_{J_K(c)}\bigr)\bmod 2,
 \end{aligned}
\end{equation}
which is the action of $g\circ h$ on $K$.  Hence
$\rho_{g\circ h}=\rho_g\circ\rho_h$ on messages and keys, and together with
$\rho_e=\mathrm{id}$ these are the first two identities of
\eqref{eq:equivariance}.  Every $g=(\sigma,f)$ has the inverse
$g^{-1}=(\sigma^{-1},f')$ with $f'_{\sigma(j)}=f_j$, that is,
$f'_k=f_{\sigma^{-1}(k)}$.  Indeed, by \eqref{eq:app-co-product} the
product $g\circ g^{-1}$ has permutation $\sigma\circ\sigma^{-1}=\mathrm{id}$
and flips $(f'_k+f_{\sigma^{-1}(k)})\bmod 2=0$ at every $k$, and
$g^{-1}\circ g$ has permutation $\mathrm{id}$ and flips
$(f_j+f'_{\sigma(j)})\bmod 2=0$ at every $j$, so both products equal $e$.
The action on messages is faithful: if $\rho_g(m)=m$ for every $m$, then
$m=0$ gives $f_j=[\rho_g(0)]_{\sigma(j)}=0$ for every $j$, and then
$[\rho_g(m)]_{\sigma(j)}=m_j$ for every $m$ forces $\sigma(j)=j$, as the
unit vector $m$ with a single one at position $j$ shows.  Therefore
$g\mapsto\rho_g$ is an injective map from $\mathcal G$ to the bijections of
$\{0,1\}^L$ that sends $\circ$ to composition, $e$ to the identity, and
$g^{-1}$ to the inverse, so $\circ$ inherits associativity from
composition and $(\mathcal G,\circ)$ is a group.  It has $2^L\,L!$
elements and is therefore finite.
Finally, for the last identity of \eqref{eq:equivariance} put $j=J_K(c)$.
Then $J_{\rho_g(K)}(c)=\sigma(j)$, and by \eqref{eq:app-co-action} the
membership condition $C_{\rho_g(K)}(c,v)=[\rho_g(m)]_{\sigma(j)}$ reads
\begin{equation}
 \bigl(C_K(c,v)+f_j\bigr)\bmod 2=(m_j+f_j)\bmod 2.
 \label{eq:app-co-set-cond}
\end{equation}
Adding the same bit $f_j$ to both sides of an equation modulo~2 does not
change its solutions, so \eqref{eq:app-co-set-cond} holds exactly when
$C_K(c,v)=m_j$.  Hence
\begin{equation}
 \begin{aligned}
 G_{\rho_g(K)}^{\rho_g(m)}(c)
 &=\bigl\{v: C_{\rho_g(K)}(c,v)=[\rho_g(m)]_{\sigma(j)}\bigr\}\\
 &=\bigl\{v: C_K(c,v)=m_j\bigr\}=G_K^m(c).
 \end{aligned}
\end{equation}
By Step~1 of Appendix~\ref{app:orbit-proof},
which used only the identities \eqref{eq:equivariance}, every $\rho_g$ is
therefore a bijection with inverse $\rho_{g^{-1}}$, and $\mathcal H_M$ is a
subgroup of $\mathcal G$.

\smallskip\noindent
\textbf{Step 2: determine the stabilizer.}
Let $g=(\sigma,f)\in\mathcal G$.  By \eqref{eq:app-co-action},
$\rho_g(m^{(i)})=m^{(i)}$ holds exactly when
$[\rho_g(m^{(i)})]_{\sigma(j)}=m^{(i)}_{\sigma(j)}$ for every $j$, that is,
when
\begin{equation}
 (m^{(i)}_j+f_j)\bmod 2=m^{(i)}_{\sigma(j)}\qquad(j\in[L]).
 \label{eq:app-co-stab}
\end{equation}
We use \eqref{eq:app-co-stab} first for the reference account and then for
the others.  For $i=1$, adding $m^{(1)}_j$ to both sides of
\eqref{eq:app-co-stab} modulo~2 and using $(u+u)\bmod 2=0$ gives
$f_j=(m^{(1)}_j+m^{(1)}_{\sigma(j)})\bmod 2$ for every $j$.  Hence the
flips of an element of $\mathcal H_M$ are determined by its permutation,
and we write
\begin{equation}
 f^{\sigma}_j:=\bigl(m^{(1)}_j+m^{(1)}_{\sigma(j)}\bigr)\bmod 2\qquad(j\in[L]).
 \label{eq:app-co-fsigma}
\end{equation}
For $i\geq2$, substitute $f=f^{\sigma}$ into \eqref{eq:app-co-stab} and add
$m^{(1)}_{\sigma(j)}$ to both sides modulo~2.  By
\eqref{eq:app-co-fsigma} and $(u+u)\bmod 2=0$, the left side becomes
\begin{equation}
 \begin{aligned}
 &\bigl(m^{(i)}_j+f^{\sigma}_j+m^{(1)}_{\sigma(j)}\bigr)\bmod 2\\
 &\quad=\bigl(m^{(i)}_j+m^{(1)}_j\bigr)\bmod 2=x_{j,i},
 \end{aligned}
\end{equation}
and the right side becomes
\begin{equation}
 \bigl(m^{(i)}_{\sigma(j)}+m^{(1)}_{\sigma(j)}\bigr)\bmod 2=x_{\sigma(j),i},
\end{equation}
where both final equalities are the definition of the signatures in (S3).
Adding the same bit to both sides of an equation modulo~2 is reversible,
so for $i\geq2$ condition \eqref{eq:app-co-stab} is equivalent to
$x_{j,i}=x_{\sigma(j),i}$.
Requiring this for every $i\in\{2,\ldots,n\}$ says $x_{\sigma(j)}=x_j$,
that is, $\sigma(j)$ is a twin of $j$.  Every step above is an
equivalence, so we have shown
\begin{equation}
 \begin{aligned}
 (\sigma,f)\in\mathcal H_M
 \iff{}& f=f^{\sigma}\text{ and}\\
 & x_{\sigma(j)}=x_j\text{ for every }j\in[L].
 \end{aligned}
 \label{eq:app-co-stab-char}
\end{equation}
Let $\Sigma_{\mathcal T}$ be the set of permutations $\sigma\in\mathfrak S_L$
with $x_{\sigma(j)}=x_j$ for every $j$.  Such a $\sigma$ maps every twin
class $T$ into itself, and since $\sigma$ is injective and $T$ is finite,
$\sigma(T)=T$.  Thus $\sigma$ restricts to a permutation $\sigma_T$ of each
class.  Conversely, every family $(\sigma_T)_{T\in\mathcal T}$ of
permutations $\sigma_T\in\mathfrak S(T)$ defines an element of
$\Sigma_{\mathcal T}$ by $\sigma(j)=\sigma_T(j)$ for $j\in T$, because the
classes partition $[L]$ and $\sigma(j)$ then has the signature of $j$.
Identifying a permutation with its tuple of restrictions, we obtain
\begin{equation}
 \begin{aligned}
 \Sigma_{\mathcal T}&=\prod_{T\in\mathcal T}\mathfrak S(T),\\
 \mathcal H_M&=\{(\sigma,f^{\sigma}):\sigma\in\Sigma_{\mathcal T}\},
 \end{aligned}
 \label{eq:app-co-stab-set}
\end{equation}
where the second line restates \eqref{eq:app-co-stab-char}, and the map
$\sigma\mapsto(\sigma,f^{\sigma})$ is a bijection from $\Sigma_{\mathcal T}$
onto $\mathcal H_M$.

\smallskip\noindent
\textbf{Step 3: pass to coordinates relative to the reference message.}
For a message $m$ define $r(m):=(m+m^{(1)})\bmod 2$ coordinatewise, so that
$R:=r(m')=(r_1,\ldots,r_L)$ is the vector of relative target bits of (S3).
The map $r$ is a bijection of $\{0,1\}^L$, since
$r(r(m))=(m+m^{(1)}+m^{(1)})\bmod 2=m$.  For $\sigma\in\mathfrak S_L$ and
a bit vector $u$, let $\pi_\sigma(u)$ be the vector whose coordinates are
those of $u$ moved by $\sigma$,
\begin{equation}
 \bigl[\pi_\sigma(u)\bigr]_{\sigma(j)}:=u_j\qquad(j\in[L]).
 \label{eq:app-co-pi}
\end{equation}
Now let $g=(\sigma,f^{\sigma})\in\mathcal H_M$.  By \eqref{eq:app-co-action},
\eqref{eq:app-co-fsigma}, and $(u+u)\bmod 2=0$,
\begin{equation}
 \begin{aligned}
 &\bigl[r(\rho_g(m'))\bigr]_{\sigma(j)}\\
 &\quad=\bigl(\bigl[\rho_g(m')\bigr]_{\sigma(j)}+m^{(1)}_{\sigma(j)}\bigr)\bmod 2\\
 &\quad=\bigl(m'_j+f^{\sigma}_j+m^{(1)}_{\sigma(j)}\bigr)\bmod 2\\
 &\quad=\bigl(m'_j+m^{(1)}_j\bigr)\bmod 2=r_j,
 \end{aligned}
 \label{eq:app-co-relative}
\end{equation}
that is, $r(\rho_g(m'))=\pi_\sigma(R)$: an element of the stabilizer moves
the relative bits of the target along its permutation and flips none of
them.  Let $\mathcal O:=\{\pi_\sigma(R):\sigma\in\Sigma_{\mathcal T}\}$.
Since $r$ is a bijection and, by \eqref{eq:app-co-stab-set}, $\sigma$
ranges over $\Sigma_{\mathcal T}$ as $g$ ranges over $\mathcal H_M$,
\begin{equation}
 \operatorname{Orb}_{\mathcal H_M}(m')=r^{-1}(\mathcal O),
 \quad
 |\operatorname{Orb}_{\mathcal H_M}(m')|=|\mathcal O|.
 \label{eq:app-co-orbit-rel}
\end{equation}

\smallskip\noindent
\textbf{Step 4: decompose along the twin classes.}
Write $u|_T\in\{0,1\}^T$ for the restriction of a bit vector $u$ to the
coordinates in $T$.  Because the classes partition $[L]$, the map
$u\mapsto(u|_T)_{T\in\mathcal T}$ is a bijection from $\{0,1\}^L$ onto
$\prod_{T}\{0,1\}^T$.  For $\sigma\in\Sigma_{\mathcal T}$ and $k\in T$ we
have $\sigma^{-1}(k)\in T$ by Step~2, so \eqref{eq:app-co-pi} gives
$[\pi_\sigma(u)]_k=u_{\sigma^{-1}(k)}=[\pi_{\sigma_T}(u|_T)]_k$, that is,
\begin{equation}
 \pi_\sigma(u)\big|_T=\pi_{\sigma_T}\bigl(u|_T\bigr)\qquad(T\in\mathcal T).
 \label{eq:app-co-restrict}
\end{equation}
Put $\mathcal O_T:=\{\pi_\tau(R|_T):\tau\in\mathfrak S(T)\}$.  We claim
that $\mathcal O$ consists exactly of the vectors $u$ with $u|_T\in\mathcal
O_T$ for every $T$.  If $u=\pi_\sigma(R)$ with $\sigma\in\Sigma_{\mathcal T}$,
then \eqref{eq:app-co-restrict} gives $u|_T=\pi_{\sigma_T}(R|_T)\in\mathcal
O_T$ for every $T$.  Conversely, if $u|_T\in\mathcal O_T$ for every $T$,
choose $\tau_T\in\mathfrak S(T)$ with $\pi_{\tau_T}(R|_T)=u|_T$, and let
$\sigma\in\Sigma_{\mathcal T}$ be the permutation with restrictions
$\sigma_T=\tau_T$, which exists by \eqref{eq:app-co-stab-set}.  Then
\eqref{eq:app-co-restrict} gives $\pi_\sigma(R)|_T=u|_T$ for every $T$, so
$\pi_\sigma(R)=u$ and $u\in\mathcal O$.  Hence
$u\mapsto(u|_T)_{T\in\mathcal T}$ maps $\mathcal O$ bijectively onto
$\prod_T\mathcal O_T$, and
\begin{equation}
 |\mathcal O|=\prod_{T\in\mathcal T}|\mathcal O_T|.
 \label{eq:app-co-product-count}
\end{equation}

\smallskip\noindent
\textbf{Step 5: count within one twin class.}
Fix $T\in\mathcal T$, put $k=|T|$, and let $w^\star=R|_T$, which has
exactly $a_T$ coordinates equal to one by the definition of $a_T$ in (S3).
We claim
\begin{equation}
 \mathcal O_T=\bigl\{w\in\{0,1\}^T:\ w\text{ has exactly } a_T\text{ ones}\bigr\}.
 \label{eq:app-co-class}
\end{equation}
Moving coordinates does not change how many of them equal one, so
$\mathcal O_T$ is contained in the right side.  Conversely, let $w$ have
exactly $a_T$ ones, and let $P=\{j\in T:w^\star_j=1\}$ and
$Q=\{j\in T:w_j=1\}$, two subsets of $T$ with $|P|=|Q|=a_T$ and hence
$|T\setminus P|=|T\setminus Q|=k-a_T$.  Choose a bijection from $P$ onto
$Q$ and a bijection from $T\setminus P$ onto $T\setminus Q$, and let $\tau$
be their union, a permutation of $T$.  For $j\in P$,
\eqref{eq:app-co-pi} gives $[\pi_\tau(w^\star)]_{\tau(j)}=w^\star_j=1$
with $\tau(j)\in Q$, and for $j\in T\setminus P$ it gives
$[\pi_\tau(w^\star)]_{\tau(j)}=0$ with $\tau(j)\notin Q$.  Since $\tau$ is
a bijection of $T$, the vector $\pi_\tau(w^\star)$ equals one exactly on
$Q$, so $\pi_\tau(w^\star)=w$ and $w\in\mathcal O_T$.  The right side of
\eqref{eq:app-co-class} is in bijection with the $a_T$-element subsets of
the $k$-element set $T$, so
\begin{equation}
 |\mathcal O_T|=\binom{|T|}{a_T}.
 \label{eq:app-co-class-count}
\end{equation}

\smallskip\noindent
\textbf{Step 6: conclude.}
Substituting \eqref{eq:app-co-class-count} into
\eqref{eq:app-co-product-count} and the result into
\eqref{eq:app-co-orbit-rel} gives
\begin{equation}
 |\operatorname{Orb}_{\mathcal H_M}(m')|=|\mathcal O|
 =\prod_{T\in\mathcal T}\binom{|T|}{a_T},
\end{equation}
which is \eqref{eq:app-co-claim}.  Together with
\eqref{eq:ambiguity-count}, this identifies the main text's $N(M,m')$ with
the number of indistinguishable messages.  As in
Step~7 of Appendix~\ref{app:orbit-proof}, the count is over distinct
messages: several elements of $\mathcal H_M$ may send $m'$ to the same
message, which is why we counted the vectors in $\mathcal O$ rather than
the permutations in $\Sigma_{\mathcal T}$.  Substituting this identity into
the general bound of Appendix~\ref{app:orbit-proof} gives
\eqref{eq:targeting-bound} and completes the proof of
Theorem~\ref{thm:targeting-bound}.
\end{proof}

\paragraph{Consequences used in the main text.}
For $k\geq1$ the binomial coefficient $\binom{k}{a}$ equals $1$ exactly
when $a\in\{0,k\}$, and a product of positive integers equals $1$ exactly
when every factor does.  Hence the orbit is trivial if and only if
$a_T\in\{0,|T|\}$ for every class, that is, if and only if the relative
bits $r_j$ are constant on every twin class.  Choosing this constant
independently for each of the $|\mathcal T|$ classes gives
$2^{|\mathcal T|}$ vectors $r$, and $m'=(r+m^{(1)})\bmod 2$ is a bijection
between these vectors and the targets with a trivial orbit, so exactly
$2^{|\mathcal T|}$ targets are reachable.  When every position has a
distinct signature, every class is a singleton and all $2^L$ targets are
reachable.  Finally, \eqref{eq:app-co-relative} is the reason behind the
construction \eqref{eq:target-set}: for a target with a trivial orbit, $r$
takes a single value $r_T$ on each class, so the attacker needs only the
class $T(c)$ of a context to decide whether the target's set at $c$ is
$G_1(c)$ or its complement.

\section{Proof of certificate validity}
\label{app:certificate-proof}

We prove Theorem~\ref{thm:certificate-validity} for a fixed observation
$O$ and rewrite $z$.  In the ideal-key model, the hidden positions
$J_K(c)$ are initially independent and uniform on $[L]$ across distinct
contexts.  The controlled-account sets observed at $c$ reveal the account
signature $x_{J_K(c)}$, and therefore the twin class $T(c)$, but not the
position inside that class.  The within-class symmetries proved in
Appendix~\ref{app:orbit-proof} assign equal posterior probability to every
position in $T(c)$.  Independence across distinct contexts is preserved
because the ideal key assigns them independently and the observations expose
only their classes, and repeated occurrences of the same context share the same
assignment.  If no class was identified, the observation gives no position
information and the posterior remains uniform on $[L]$.  Consequently, the
conditional law of the position map $J_K$ on the contexts of $z$ is exactly
$\Pi_O$.

For a key $K$, define the true target vote of a scored pair by
\begin{equation}
 \sigma_K(c,v):=
 \begin{cases}
  +1,&v\in G_K^{m'}(c),\\
  -1,&v\notin G_K^{m'}(c).
 \end{cases}
 \label{eq:app-true-certificate-vote}
\end{equation}
When the target set at $c$ was recovered, the attacker's vote
$\sigma(c,v)$ equals $\sigma_K(c,v)$.  When it was not recovered, the
attacker sets $\sigma(c,v)=-1$, so in every case
\begin{equation}
 \sigma(c,v)\leq \sigma_K(c,v).
 \label{eq:app-conservative-vote}
\end{equation}
For the true assignment $J_K$, let
\begin{equation}
 M_j^K:=\sum_{(c,v)\in\mathcal P(z):\,J_K(c)=j}\sigma_K(c,v)
 \label{eq:app-true-certificate-margin}
\end{equation}
be the decoder's target margin at position $j$.  Equations
\eqref{eq:app-conservative-vote} and
\eqref{eq:app-true-certificate-margin} give
\begin{equation}
 M_j(J_K)\leq M_j^K\qquad(j\in[L]).
 \label{eq:app-margin-domination}
\end{equation}
Therefore, if $E$ holds for the conservative tally, every position is
covered and has $M_j^K>0$.  The strict-majority decoder must then output
$m'$ at every position.  Pointwise in the hidden key,
\begin{equation}
 \mathbf 1\{E(J_K)\}
 \leq \mathbf 1\{D_K(z)=m'\}.
 \label{eq:app-certificate-indicators}
\end{equation}
Taking conditional expectations and using that
$J_K\mid O\sim\Pi_O$ yields
\begin{equation}
 \begin{aligned}
 \operatorname{Cert}(z,m'\mid O)
 &=\Pr_{J\sim\Pi_O}[E]\\
 &=\Pr[E(J_K)\mid O]\\
 &\leq\Pr[D_K(z)=m'\mid O],
 \end{aligned}
\end{equation}
which proves the lower bound.  If every target set used by $z$ was recovered,
then \eqref{eq:app-conservative-vote} is an equality for every pair.  The
strict-majority decoder outputs $m'$ exactly when every position is covered
and has positive target margin, so
\eqref{eq:app-certificate-indicators} is also an equality.  This proves the
equality claim.

\section{The vote bound}
\label{app:theory}

The certificate of Section~\ref{sec:certificate} evaluates the decoder's
vote tally exactly under the attacker's posterior.  The following bound
complements it from the other side: it lower-bounds the success probability
in terms of the true vote counts and vote biases, which the attacker cannot
see, and explains why the constrained rule succeeds once enough contexts
carry a target vote.  We fix the setting with numbered assumptions, state
the bound, and prove it in four steps: the vote probabilities of the two
rewriting rules, the decoder's success as an intersection of per-position
majority events, Hoeffding's inequality for each failure event, and a union
bound over the positions.

\paragraph{Setting.}
Fix a source $y$, a target $m'\in\{0,1\}^L$, and a rewrite
$z=(z_1,\ldots,z_T)$ of $y$ produced by the rule of
Section~\ref{sec:constructive-attack} for any threshold $\tau\geq0$, with
or without the 4-gram block, and assume the following.
\begin{enumerate}[label=\textnormal{(V\arabic*)},leftmargin=*,labelsep=0.3em]
\item \textbf{Decoder.}
The decoder is the binary position--color decoder of
\eqref{eq:binary-label-set}.  Let $\mathcal P(z)$ be the set of distinct
context--token pairs $(c_t,z_t)$ of $z$, each pair counted once however
often it occurs.  Position $j\in[L]$ receives the pairs with $J_K(c)=j$,
and for a color $b\in\{0,1\}$ its vote count is
\begin{equation}
 \begin{aligned}
 V_j(b):=\#\bigl\{(c,v)\in\mathcal P(z):&J_K(c)=j,\\
 &C_K(c,v)=b\bigr\},
 \end{aligned}
 \label{eq:app-vb-votes}
\end{equation}
so that $N_j:=V_j(0)+V_j(1)$ is the number of pairs at position $j$.  The
decoder outputs at position $j$ the color $b$ with $V_j(b)>V_j(1-b)$
whenever such a color exists, and $D_K(z)=m'$ requires this output to be
$m'_j$ at every position.
\item \textbf{Rewrite rule and recovered sets.}
The orbit of $m'$ is trivial and the attacker has recovered every target
set $G_K^{m'}(c)$ exactly, so that at every step $t$ the target-compatible
set is $B_t=G_K^{m'}(c_t)$, since the decoder
of (V1) rewards membership in the target's set.  Token $z_t$ is drawn from
the constrained distribution \eqref{eq:hard-attack} when $q_t\geq\tau$ and
$B_t$ contains an available token, and otherwise from the soft bias
\eqref{eq:soft-attack}, where $p_t$ is the paraphraser's distribution at
step $t$, after the 4-gram block has removed its tokens if the block is
in use, and $q_t=\sum_{u\in B_t}p_t(u)$.
\item \textbf{Vote variables.}
Enumerate the pairs of position $j$ as $(c_{j,s},v_{j,s})$ for
$s=1,\ldots,N_j$, and put
\begin{equation}
 \begin{aligned}
 X_{j,s}&:=\mathbf 1\{C_K(c_{j,s},v_{j,s})=m'_j\}\\
 &\phantom{:}=\mathbf 1\{v_{j,s}\in G_K^{m'}(c_{j,s})\},
 \end{aligned}
 \label{eq:app-vb-x}
\end{equation}
the two forms being equal by \eqref{eq:binary-label-set} because
$J_K(c_{j,s})=j$.
\item \textbf{Independence and bias.}
Conditionally on the vote counts $(N_j)_{j\in[L]}$, the variables
$X_{j,s}$ are independent, and there are numbers $\rho_j>1/2$ with
$\Pr[X_{j,s}=1]\geq\rho_j$ for every $j$ and $s$.  All probabilities
below are conditional on $(N_j)_{j\in[L]}$.
\end{enumerate}
\begin{theorem}[Vote bound]
\label{thm:vote-certificate}
Under (V1) to (V4), a pair drawn by \eqref{eq:hard-attack} has $X_{j,s}=1$
with certainty, a pair drawn by \eqref{eq:soft-attack} at step $t$ has
$X_{j,s}=1$ with probability
$q_te^{\delta_{\mathrm{atk}}}/(q_te^{\delta_{\mathrm{atk}}}+1-q_t)$, and
\begin{equation}
 \Pr[D_K(z)=m']\geq 1-\sum_{j=1}^{L}e^{-2N_j(\rho_j-1/2)^2}.
 \label{eq:vote-bound}
\end{equation}
\end{theorem}

\begin{proof}
\smallskip\noindent
\textbf{Step 1: compute the vote probabilities of the two rules.}
Let $t$ be a step whose pair $(c_t,z_t)$ belongs to position $j$, so that
$J_K(c_t)=j$ and $B_t=G_K^{m'}(c_t)$ by (V2).  If $z_t$ is drawn from
\eqref{eq:hard-attack}, then $r_t(v)=0$ for every $v\notin B_t$, so
$z_t\in B_t$ with probability one and the pair's vote variable equals
$\mathbf 1\{z_t\in G_K^{m'}(c_t)\}=1$.  If $z_t$ is drawn from
\eqref{eq:soft-attack}, then $r_t(v)=p_t(v)e^{\delta_{\mathrm{atk}}}/Z_t$
for $v\in B_t$ and $r_t(v)=p_t(v)/Z_t$ for $v\notin B_t$, with the
normalizer
\begin{equation}
 \begin{aligned}
 Z_t&=\sum_{v\in B_t}p_t(v)e^{\delta_{\mathrm{atk}}}+\sum_{v\notin B_t}p_t(v)\\
 &=q_te^{\delta_{\mathrm{atk}}}+1-q_t,
 \end{aligned}
\end{equation}
because $\sum_{v\notin B_t}p_t(v)=1-q_t$.  Summing $r_t$ over $B_t$,
\begin{equation}
 \Pr[z_t\in B_t]
 =\frac{q_te^{\delta_{\mathrm{atk}}}}{q_te^{\delta_{\mathrm{atk}}}+1-q_t},
 \label{eq:app-vb-soft}
\end{equation}
which is the stated probability.  When the 4-gram block removes tokens,
$p_t$ denotes the blocked and renormalized distribution, and both
computations apply unchanged with its $q_t$.  Since
$e^{\delta_{\mathrm{atk}}}>1$, the probability \eqref{eq:app-vb-soft}
exceeds $q_t$, and it exceeds $1/2$ exactly when
$q_te^{\delta_{\mathrm{atk}}}>1-q_t$.

\smallskip\noindent
\textbf{Step 2: express the decoder's success through the vote sums.}
Fix a position $j$ and put $S_j:=\sum_{s=1}^{N_j}X_{j,s}$.  Every pair at
position $j$ has one of the two colors, so \eqref{eq:app-vb-votes} and
\eqref{eq:app-vb-x} give
\begin{equation}
 V_j(m'_j)=S_j,\qquad V_j(1-m'_j)=N_j-S_j.
\end{equation}
By (V1), the decoder outputs $m'_j$ at position $j$ whenever
$V_j(m'_j)>V_j(1-m'_j)$, that is, whenever $S_j>N_j-S_j$, which is
equivalent to $S_j>N_j/2$.  Define the failure event
\begin{equation}
 F_j:=\{S_j\leq N_j/2\}.
 \label{eq:app-vb-fail}
\end{equation}
If $N_j=0$, then $S_j=0$ and $F_j$ is certain, which matches the fact that
an empty position has no majority.  If no $F_j$ occurs, every position
outputs $m'_j$ and $D_K(z)=m'$.  Hence
\begin{equation}
 \begin{aligned}
 \{D_K(z)=m'\}&\supseteq\bigcap_{j=1}^{L}F_j^{\mathrm c},\\
 \Pr[D_K(z)=m']&\geq 1-\Pr\Bigl[\bigcup_{j=1}^{L}F_j\Bigr].
 \end{aligned}
 \label{eq:app-vb-success}
\end{equation}
The inclusion is an equality when the decoder never outputs $m'_j$ from a
tie, which the strict majority rule of (V1) guarantees, but the inequality
is all the proof needs.

\smallskip\noindent
\textbf{Step 3: bound each failure event by Hoeffding's inequality.}
Hoeffding's inequality \citep{hoeffding1963probability} states that for
independent random variables $Y_1,\ldots,Y_N$ with $0\leq Y_s\leq1$ and
$S=\sum_{s=1}^{N}Y_s$,
\begin{equation}
 \Pr\bigl[S-\mathbb E[S]\leq -u\bigr]\leq e^{-2u^2/N}\qquad(u>0).
 \label{eq:app-vb-hoeffding}
\end{equation}
Fix $j$ with $N_j\geq1$.  By (V4) the $X_{j,s}$ are independent, take
values in $\{0,1\}$, and satisfy
\begin{equation}
 \mathbb E[S_j]=\sum_{s=1}^{N_j}\Pr[X_{j,s}=1]\geq N_j\rho_j.
 \label{eq:app-vb-mean}
\end{equation}
On $F_j$ we have $S_j\leq N_j/2$, so by \eqref{eq:app-vb-mean}
\begin{equation}
 S_j-\mathbb E[S_j]\leq\frac{N_j}{2}-N_j\rho_j=-N_j\Bigl(\rho_j-\frac12\Bigr),
\end{equation}
and $u:=N_j(\rho_j-1/2)$ is positive because $\rho_j>1/2$.  Applying
\eqref{eq:app-vb-hoeffding} with $N=N_j$ and this $u$,
\begin{equation}
 \begin{aligned}
 \Pr[F_j]&\leq\Pr\bigl[S_j-\mathbb E[S_j]\leq -u\bigr]\\
 &\leq e^{-2u^2/N_j}=e^{-2N_j(\rho_j-1/2)^2}.
 \end{aligned}
 \label{eq:app-vb-perpos}
\end{equation}
If $N_j=0$, then $\Pr[F_j]=1=e^{0}$, so \eqref{eq:app-vb-perpos} holds
for every position.

\smallskip\noindent
\textbf{Step 4: combine the positions.}
The union bound gives
$\Pr[\bigcup_{j}F_j]\leq\sum_{j=1}^{L}\Pr[F_j]$, and inserting
\eqref{eq:app-vb-perpos} into \eqref{eq:app-vb-success} yields
\begin{equation}
 \Pr[D_K(z)=m']\geq 1-\sum_{j=1}^{L}e^{-2N_j(\rho_j-1/2)^2},
\end{equation}
which is \eqref{eq:vote-bound}.  The bound is
informative only when the sum is below one, and it needs no assumption on
how the decoder resolves ties, because Step~2 counted a tie as a failure.
\end{proof}

\paragraph{Reading the bound.}
Under the rule with $\tau=0$ and every set recovered, Step~1 gives
$X_{j,s}=1$ for every pair, so $\rho_j=1$ is admissible and the bound reads
$1-\sum_j e^{-N_j/2}$, a smooth relaxation of the exact statement that the
decoder then succeeds as soon as every position receives a pair.  Under the
fallback \eqref{eq:soft-attack}, $\rho_j$ can be taken as the smallest
probability \eqref{eq:app-vb-soft} over the fallback steps of position $j$,
provided it exceeds $1/2$, and the bound then says how many votes a
position needs before its majority is reliable.  The attacker cannot
evaluate \eqref{eq:vote-bound} itself, because it does not know the
positions $J_K(c)$ and hence the counts $N_j$, which is why
Section~\ref{sec:certificate} certifies a rewrite from the posterior over
the twin classes instead.

\section{Worked examples of relabelings and stabilizers}
\label{app:relabeling-examples}

This appendix works through two relabelings of a binary position--color
watermark.  The first survives after the controlled-account messages are
known and therefore lies in $\mathcal H_M$, while the second is a valid symmetry of
the watermark, but is ruled out by those messages.  Both are pure position
relabelings: their per-position color-flip vector is zero.

\subsection{Common setup}

Let the message space be $\mathcal M=\{0,1\}^3$ and the vocabulary be
$\mathcal V=\{a,b,c,d\}$.  At every context $c$, the key selects a message
position $J_K(c)\in\{1,2,3\}$ and assigns each token a color
$C_K(c,v)\in\{0,1\}$.  The green list is
\begin{equation}
 G_K^m(c)
 =\{v\in\mathcal V:C_K(c,v)=m_{J_K(c)}\}.
 \label{eq:app-worked-preferred-set}
\end{equation}
For a concrete context, suppose
\begin{equation}
 J_K(c)=2,\qquad m=(0,1,1),
 \label{eq:app-worked-context}
\end{equation}
and let the token colors be
\begin{center}
\small
\renewcommand{\arraystretch}{1.15}
\begin{tabular}{@{}ccccc@{}}
\toprule
token $v$ & $a$ & $b$ & $c$ & $d$ \\
\midrule
$C_K(c,v)$ & $0$ & $0$ & $1$ & $1$ \\
\bottomrule
\end{tabular}
\end{center}
The key reads $m_{J_K(c)}=m_2=1$, so
\begin{equation}
 G_K^m(c)=\{v:C_K(c,v)=1\}=\{c,d\}.
 \label{eq:app-worked-concrete-set}
\end{equation}
Finally, suppose that the attacker controls two accounts with messages
\begin{equation}
 m^{(1)}=(0,1,1),\qquad m^{(2)}=(1,0,0),
 \label{eq:app-worked-accounts}
\end{equation}
and write $M=(m^{(1)},m^{(2)})$.

\subsection{A relabeling retained by the stabilizer}

Let $g$ swap the names of positions $2$ and $3$.  Write this permutation as
\begin{equation}
 \sigma_g(1)=1,\qquad \sigma_g(2)=3,\qquad \sigma_g(3)=2.
 \label{eq:app-g-permutation}
\end{equation}

\begin{enumerate}[label=\textbf{Step \arabic*.},leftmargin=*]
\item \textbf{Relabel the message.}
The bit formerly called position $j$ must appear at the new position
$\sigma_g(j)$.  Hence
\begin{equation}
 \begin{aligned}
  \bigl[\rho_g(m)\bigr]_{\sigma_g(j)}&=m_j,\\
  \rho_g(m_1,m_2,m_3)&=(m_1,m_3,m_2).
 \end{aligned}
 \label{eq:app-g-message}
\end{equation}
For example, $(0,1,0)$ becomes $(0,0,1)$.

\item \textbf{Relabel the key.}
The key must follow the same renaming of positions, while the token colors
remain unchanged:
\begin{equation}
 \begin{aligned}
  J_{\rho_g(K)}(c)&=\sigma_g(J_K(c)),\\
  C_{\rho_g(K)}(c,v)&=C_K(c,v).
 \end{aligned}
 \label{eq:app-g-key}
\end{equation}
In the concrete context of \eqref{eq:app-worked-context}, the new key therefore
selects position $3$.

\item \textbf{Check the concrete green list.}
Here $\rho_g(m)=(0,1,1)$ and
\begin{equation}
 \begin{aligned}
  \bigl[\rho_g(m)\bigr]_{J_{\rho_g(K)}(c)}
  &=\bigl[\rho_g(m)\bigr]_3\\
  &=1=m_2.
 \end{aligned}
\end{equation}
The colors have not changed, so the relabeled system again selects
$\{c,d\}$.

\item \textbf{Verify the relabeling for every key, message, and context.}
Fix arbitrary $K,m,c$ and put $j=J_K(c)$.  Equations
\eqref{eq:app-g-message}--\eqref{eq:app-g-key} give
\begin{equation}
 \begin{aligned}
 \bigl[\rho_g(m)\bigr]_{J_{\rho_g(K)}(c)}
 &= \bigl[\rho_g(m)\bigr]_{\sigma_g(j)}\\
 &=m_j=m_{J_K(c)}.
 \end{aligned}
 \label{eq:app-g-selected-bit}
\end{equation}
Substituting this identity and the unchanged colors into
\eqref{eq:app-worked-preferred-set} yields
\begin{equation}
 \begin{aligned}
 G_{\rho_g(K)}^{\rho_g(m)}(c)
 &=\left\{v:
   \substack{C_{\rho_g(K)}(c,v)\\
   =\,[\rho_g(m)]_{J_{\rho_g(K)}(c)}}\right\}\\
 &=\{v:C_K(c,v)=m_{J_K(c)}\}\\
 &=G_K^m(c).
 \end{aligned}
 \label{eq:app-g-equivariance}
\end{equation}
Moreover, $\sigma_g$ is its own inverse.  Thus $g$ is an invertible,
set-preserving relabeling, so $g\in\mathcal G$.

\item \textbf{Check the controlled messages.}
Both controlled messages are fixed because their second and third bits agree:
\begin{center}
\small
\renewcommand{\arraystretch}{1.15}
\begin{tabular}{@{}ccc@{}}
\toprule
account & $m^{(i)}$ & $\rho_g(m^{(i)})$ \\
\midrule
$1$ & $(0,1,1)$ & $(0,1,1)$ \\
$2$ & $(1,0,0)$ & $(1,0,0)$ \\
\bottomrule
\end{tabular}
\end{center}
Therefore $\rho_g(m^{(i)})=m^{(i)}$ for every controlled account, which is
exactly the stabilizer condition in \eqref{eq:stabilizer}.  Consequently,
\begin{equation}
 \boxed{g\in\mathcal G\quad\text{and}\quad g\in\mathcal H_M.}
 \label{eq:app-g-conclusion}
\end{equation}
\end{enumerate}

\subsection{A relabeling excluded by the stabilizer}

Now let $h$ swap the names of positions $1$ and $2$:
\begin{equation}
 \sigma_h(1)=2,\qquad \sigma_h(2)=1,\qquad \sigma_h(3)=3.
 \label{eq:app-h-permutation}
\end{equation}

\begin{enumerate}[label=\textbf{Step \arabic*.},leftmargin=*]
\item \textbf{Relabel the message.}
As before, the old bit $m_j$ is moved to position $\sigma_h(j)$:
\begin{equation}
 \begin{aligned}
  \bigl[\rho_h(m)\bigr]_{\sigma_h(j)}&=m_j,\\
  \rho_h(m_1,m_2,m_3)&=(m_2,m_1,m_3).
 \end{aligned}
 \label{eq:app-h-message}
\end{equation}

\item \textbf{Relabel the key.}
Apply the same position renaming to the key and leave its colors unchanged:
\begin{equation}
 \begin{aligned}
  J_{\rho_h(K)}(c)&=\sigma_h(J_K(c)),\\
  C_{\rho_h(K)}(c,v)&=C_K(c,v).
 \end{aligned}
 \label{eq:app-h-key}
\end{equation}
Thus, in the concrete context, the old choice $J_K(c)=2$ becomes
$J_{\rho_h(K)}(c)=1$.

\item \textbf{Check the concrete green list.}
The message becomes $\rho_h(m)=(1,0,1)$, but the new key reads
\begin{equation}
 \begin{aligned}
  \bigl[\rho_h(m)\bigr]_{J_{\rho_h(K)}(c)}
  &=\bigl[\rho_h(m)\bigr]_1\\
  &=1=m_2.
 \end{aligned}
\end{equation}
The selected bit and token colors are therefore unchanged, so the green
list is again $\{c,d\}$.

\item \textbf{Verify the relabeling for every key, message, and context.}
For arbitrary $K,m,c$, put $j=J_K(c)$.  Then
\begin{equation}
 \begin{aligned}
 \bigl[\rho_h(m)\bigr]_{J_{\rho_h(K)}(c)}
 &= \bigl[\rho_h(m)\bigr]_{\sigma_h(j)}\\
 &=m_j=m_{J_K(c)}.
 \end{aligned}
\end{equation}
Together with $C_{\rho_h(K)}(c,v)=C_K(c,v)$, this gives, for every $K,m,c$,
\begin{equation}
 G_{\rho_h(K)}^{\rho_h(m)}(c)=G_K^m(c)
 \label{eq:app-h-equivariance}
\end{equation}
The permutation $\sigma_h$ is also its own inverse, so $h$ is a valid
relabeling and $h\in\mathcal G$.

\item \textbf{Check the controlled messages.}
This time neither controlled message is fixed:
\begin{center}
\small
\renewcommand{\arraystretch}{1.15}
\begin{tabular}{@{}ccc@{}}
\toprule
account & $m^{(i)}$ & $\rho_h(m^{(i)})$ \\
\midrule
$1$ & $(0,1,1)$ & $(1,0,1)$ \\
$2$ & $(1,0,0)$ & $(0,1,0)$ \\
\bottomrule
\end{tabular}
\end{center}
Already $\rho_h(m^{(1)})\ne m^{(1)}$, so $h$ fails the defining condition of
$\mathcal H_M$.  Hence
\begin{equation}
 \boxed{h\in\mathcal G\quad\text{but}\quad h\notin\mathcal H_M.}
 \label{eq:app-h-conclusion}
\end{equation}
\end{enumerate}

\subsection{What the two examples show}

\begin{center}
\small
\renewcommand{\arraystretch}{1.2}
\begin{tabular}{@{}clcc@{}}
\toprule
map & \shortstack{position\\renaming} & in $\mathcal G$? & in $\mathcal H_M$? \\
\midrule
$g$ & $2\leftrightarrow3$ & yes & yes \\
$h$ & $1\leftrightarrow2$ & yes & no \\
\bottomrule
\end{tabular}
\end{center}
Thus $\mathcal G$ contains all valid symmetries of the watermark, whereas
$\mathcal H_M$ contains only those symmetries that the known controlled-account
messages do not rule out.  These examples use pure position permutations.  The
full binary position--color group also permits a position permutation combined
with per-position color flips, and such a combination is a different group
element.  Therefore $h\notin\mathcal H_M$ excludes the particular pure swap
defined above, not every relabeling whose position component swaps $1$ and $2$.

\section{Use of AI assistants}
\label{app:ai-assistants}
We used large language models to help polish the presentation and writing
of this manuscript and to assist in drafting portions of the implementation
code.  All substantive research decisions and core technical contributions,
however, were made by the authors.

\end{document}